\documentclass[final,3p]{elsarticle}
\usepackage{graphicx}
\usepackage{amssymb}
\usepackage{amsthm}
\usepackage{amsmath}
\usepackage{lineno}
\usepackage{multirow}
\usepackage{caption}
\usepackage{mathrsfs}
\usepackage{sectsty}
\usepackage{epsfig}
\usepackage{url}
\newtheorem{lemma}{Lemma}

\newtheorem{corollary}{Corollary}
\newtheorem{proposition}{Proposition}
\newdefinition{assumption}{Assumption}
\newdefinition{rmk}{Remark}
\newdefinition{algorithm}{Algorithm}
\newdefinition{problem}{Problem}
\newdefinition{definition}{Definition}
\newdefinition{appx}{Appendix}
\newproof{pf}{Proof}
\DeclareMathOperator{\tr}{Tr}
\DeclareMathOperator{\supp}{supp}
\DeclareMathOperator{\pr}{Pr}

\DeclareMathOperator*{\argmin}{arg\,min}
\usepackage{savesym}
\savesymbol{AND}
\usepackage{algorithmic}
\usepackage{algorithm}
\makeatletter
\def\ps@pprintTitle{%
  \let\@oddhead\@empty
  \let\@evenhead\@empty
  \let\@oddfoot\@empty
  \let\@evenfoot\@oddfoot
}
\makeatother
\begin{document}
\begin{frontmatter}
\title{On-line Bayesian parameter estimation in general non-linear state-space models: A  tutorial and new results\tnoteref{t1}}
\author[UofA]{Aditya Tulsyan}\ead{{tulsyan@ualberta.ca}}
\author[UofA]{Biao Huang}\ead{biao.huang@ualberta.ca}
\author[UBC,MIT]{R. Bhushan Gopaluni}\ead{bhushan.gopaluni@ubc.ca}
\author[UofA]{J. Fraser Forbes}\ead{fraser.forbes@ualberta.ca}
\address[UofA]{Department of Chemical and Materials Engineering, University of Alberta, Edmonton, AB T6G 2G6, Canada.}
\address[UBC]{Department of Chemical and Biological Engineering, University of British Columbia, Vancouver, BC V6T 1Z3, Canada.}
\address[MIT]{Department of Chemical Engineering, Massachusetts Institute of Technology, Cambridge, MA 02139, USA.}
\tnotetext[t1]{A condensed version of this article has been published in: Tulsyan, A., Huang, B., Gopaluni, R.B., Forbes, J.F. ``On simultaneous on-line state and parameter estimation in non-linear state-space models". {Journal of Process Control}, vol 23, no. 4, 2013. Authors' addresses: A. Tulsyan, B. Huang and J.F. Forbes are with the Computer Process Control Group, Department of Chemical and Materials Engineering, University of Alberta, Edmonton T6G-2G6, Alberta, Canada, (e-mail: \{tulsyan; biao.huang; fraser.forbes\}@ualberta.ca); and R.B. Gopaluni is with the Process Modeling and Control Lab, Department of Chemical and Biological Engineering, University of British Columbia, Vancouver V6T-1Z3, BC, Canada, (e-mail: bhushan.gopaluni@ubc.ca).}

\begin{abstract}
On-line estimation plays an important role in process control and monitoring. Obtaining a theoretical solution to the simultaneous state-parameter estimation problem for non-linear stochastic systems involves solving complex multi-dimensional integrals that are not amenable to analytical solution. While basic sequential Monte-Carlo (SMC) or particle filtering (PF) algorithms for simultaneous estimation exist, it is well recognized that there is a need for making these on-line algorithms non-degenerate, fast and applicable to processes with missing measurements. To overcome the deficiencies in traditional algorithms, this work proposes a Bayesian approach to on-line state and parameter estimation. Its extension to handle missing data in real-time is also provided. The simultaneous estimation is performed by filtering an extended vector of states and parameters using an adaptive sequential-importance-resampling (SIR) filter with a kernel density estimation method. The approach uses an on-line optimization algorithm based on Kullback-Leibler (KL) divergence to allow adaptation of the SIR filter for combined state-parameter estimation. An optimal tuning rule to control the width of the kernel and the variance of the artificial noise added to the parameters is also proposed. The approach is illustrated through numerical examples.
\end{abstract}
\end{frontmatter}
\section{Introduction}
\label{sec:Chap1S1}
Recent advances in high speed computation have allowed the process industries to use complex high-fidelity non-linear dynamic models, such as in: a fermentation bioreactor \cite{S2010}; polymerization \cite{AK1992}; and petroleum reservoirs \cite{E2007}. Implementing advanced control strategies or monitoring process behaviour require real-time data processing for on-line estimation of the key process states and model parameters, which are either unmeasured or unknown. An extensive literature is available on on-line state estimation using sub-optimal Bayesian filters, such as extended Kalman filters (EKFs), unscented Kalman filters (UKFs), approximate grid-based filters (GBFs), and particle filters (PFs) \cite{S1998,A2002,RB2006}; however, their extension to on-line state-parameter estimation has received attention only recently. 

In the past 15 years, several algorithms have been proposed to solve the simultaneous state-parameter estimation problem in real-time using likelihood and Bayesian derived methods. Despite the advances in SMC methods, which provide a good approximation to the optimal non-linear filter under weak assumption, simultaneous state-parameter estimation is a long-standing problem \cite{ADT2005}. This is due to the non-trivial complexities introduced with on-line estimation of the unknown model parameters \cite{C2005}. This paper considers simultaneous on-line state-parameter estimation in non-linear stochastic systems under the Bayesian framework. The existing and current developments in both Bayesian and likelihood based methods for on-line state and parameter estimation are first briefly reviewed. An exposition of parameter estimation using Bayesian and likelihood based methods can be found in \cite{Kantas2009}.

The central idea of simultaneous on-line Bayesian estimators is certainly not new. A customary approach involves selecting a prior distribution for the model parameters followed by augmenting it with the states to form an extended state vector \cite{K1998}. Theoretically, it casts the simultaneous state and parameter estimation problem into a unified filtering framework; however, due to lack of ergodicity and exponential forgetting of the joint state-parameter filter, coupled with successive resampling steps, employing this approach with any standard SMC algorithm often results in parameter sample degeneracy \cite{ADT2005,DT2003}. In other words, SMC approximation of the marginalized parameter posterior distribution is represented by a single Dirac delta function. It also causes error accumulation in successive Monte Carlo (MC) steps, which in terms of $L^p$ norm, grows exponentially or polynomially in time \cite{Kantas2009}.

A pragmatic approach to reduce parameter sample degeneracy and error accumulation in successive MC approximations is to introduce diversity to the parameter samples. This is done by adding artificial dynamics to the parameters ({e.g.,} random walk) in the extended state vector \cite{K1998,H2001}. In practice, artificial dynamics approach (ADA) has been implemented with several on-line Bayesian estimators (auxiliary SIR filter (ASIR) \cite{LW2001}, Rao-Blackwellised particle filter (RBPF) \cite{GS2003}). While this approach reduces parameter sample degeneracy and error accumulation in successive MC steps, adding artificial dynamics to the parameters, often results in over-dispersed posteriors, which is also commonly referred to as the variance inflation problem \cite{LW2001}. To overcome the posterior variance inflation problem, a kernel density estimation method is proposed  in \cite{LW2001,W1993}, in which the degenerated approximation of the marginalized parameter posterior distribution is substituted by a kernel approximation (e.g., Gaussian or Epanechnikov). The artificial dynamics approach together with kernel density estimation method efficiently introduces parameter sample diversity and can be used for state-parameter estimation in general non-linear state-space models (SSMs) with non-Gaussian noise; however, there are several limitations of this approach as summarized in \cite{Kantas2009}: (a) transforming the problem by adding artificial noise modifies the original problem, so that, it becomes hard to quantify the bias introduced in the resulting parameter estimates; and (b) the dynamics of the parameters are related to the width of the kernel and the variance of the artificial noise, which are often difficult to fine tune. For the first issue, \cite{ABBF2012} proposed the use of posterior Cram\'{e}r-Rao lower bound (PCRLB) \cite{T1998} as a benchmark for error analysis of the parameter estimates obtained using the artificial approach; whereas, for the second issue, no practical solution exists. 

The authors in \cite{C2005} used an ASIR filter for on-line state-parameter estimation with a priori knowledge based kernel width tuning rule. Compared to the SIR filter, an ASIR filter is a one-step look-ahead filter which offers an advantage by allowing importance sampling from the high likelihood region \cite{P1999}; however, the superiority of ASIR to SIR is case dependent \cite{JD2008}. Most importantly, poor performance of ASIR filter for systems with large process noise \cite{R2004} coupled with higher computational cost (compared to the SIR filter) \cite{ZC2003}, often renders it impractical for on-line applications.

The Resample-Move is an alternate on-line Bayesian estimation approach which introduces parameter sample diversity through a Markov chain Monte Carlo (MCMC) step \cite{G1998,GB2001,LC2002,C2002}. To avoid increase in the memory requirements with the MCMC step, use of a fixed dimensional sufficient statistics has also been proposed in the on-line Bayesian parameter estimation context \cite{A1999}. As opposed to the methods based on kernel or artificial dynamics, Resampling-Move algorithm has the advantage of introducing diversity without perturbing the joint state-parameter target distribution. Unfortunately, MCMC/sufficient statistics based algorithms  are known to result in approximation errors, which accumulate at least quadratically in time \cite{Kantas2009,A1999}. This problem has also been illustrated in \cite{A2004} using a sufficient statistics method. Finally, unlike the ADA, applicability of the Resample-Move approach is restricted to a certain class of low dimensional non-linear models, for some of which, tractable solution to the estimation problem is also available \cite{DM2002,S2002}. 

Apart from the developments in Bayesian estimation, maximum likelihood (ML) based algorithms for on-line parameter estimation is also an active area of research. Unlike the Bayesian estimators, where the focus is on the simultaneous state-parameter estimation, ML based methods are primarily focussed on solving the parameter estimation problem. A standard approach to on-line ML parameter estimation is the gradient method. The gradient method requires recursive computation of the likelihood of the measurements and its gradient with respect to the parameters, which is also referred to as the score function. Other than in simple models, such as in linear SSMs with Gaussian noise \cite{KS1992} or in finite state-space hidden Markov models (HMMs) \cite{LH2002}, it is impossible to exactly solve the likelihood and the score functions \cite{P2011}, and one has to resort to the use of some suitable approximations. In \cite{P2011,P2005}, use of SMC methods to approximate the likelihood and score functions for estimation using on-line gradient method is proposed. As pointed in \cite{Kantas2009,A2004}, for large dimensional problems, gradient approach scales poorly in terms of its components. 

An alternate ML approach is the on-line expectation maximization (EM) algorithm, which unlike on-line gradient method, is known to be numerically more stable \cite{A2004}. Unfortunately, like the gradient method, on-line EM algorithm can be implemented exactly only in linear SSMs with Gaussian noise \cite{E2002} and in finite state-space HMMs \cite{OC2011}. Recently, SMC based on-line EM algorithm for parameter estimation in changepoint models \cite{Y2012}, and in certain classes of the non-linear SSMs \cite{DDS2009,C2009}, for which the likelihood function belongs to the exponential family of distributions have appeared. Both on-line gradient and EM algorithms have computational complexity, which is quadratic in the number of particles used in the SMC approximation of the densities of interest. To develop computationally cheaper versions of the algorithm, pseudo on-line EM method for finite state-space HMMs \cite{R1997} and for non-linear SSMs \cite{ADT2005} have been proposed. Compared to the on-line gradient and EM algorithm, the pseudo on-line EM algorithm is computationally lighter, but fails to yield asymptotically efficient (unbiased and minimum variance) estimates. Finally, the pseudo on-line EM algorithm requires the stationary distribution of the states, which may not be always known in practice. 
\begin{table*}[t]
\caption{Summary of the Bayesian and likelihood based methods for on-line state-parameter estimation (adapted from \protect\cite{Kantas2009}). In this table, $N$ is the number of particles used in SMC approximation, $T$ is the final sampling time, and $L$ is the number of measurements in each block of data (see \protect\cite{ADT2005} for further details).}
\centering
\begin{tabular}{|c|c|c|c|c|}
\hline 
Method & Pros & Cons & Comp. cost \\ 
\hline 
Artificial Dynamics & Standard SMC applicable & Distribution altered & $\mathcal{O}(NT)$ \\ 
(Bayesian)  & No optimization involved & Difficult to tune dynamics & • \\ 
\hline 
Resample-Move  & Distribution unaltered & Restricted model class & $\mathcal{O}(NT)$ \\ 
 
(Bayesian) & No optimization involved & Degeneracy problem &\\ 

• &   & Scalability issues & \\ 
\hline 
On-line Gradient  & Asymptotically efficient & Locally optimal & $\mathcal{O}(N^2)$ \\ 
(ML) & Generally applicable & Scalability issues & per update \\ 
• & •  & Expensive & • \\ 
\hline 
On-line EM & Asymptotically efficient & Locally optimal & $\mathcal{O}(N^2)$ \\ 
(ML) & • & Restricted model class & per update\\ 
• & • & Expensive & • \\ 
\hline 
On-line EM pseudo & Minimal tuning & Needs stationary distribution & $\mathcal{O}(NL)$ \\  
(ML) & No degeneracy for small L & Loss of efficiency & per update \\ 
\hline 
\end{tabular} 
\label{T0}
\end{table*}

The key advantage of using ML estimators, such as on-line gradient \cite{P2011,P2005} and EM algorithm \cite{DDS2009,Y2012} is that these methods yield asymptotically efficient estimates, at least in theory; however, in many situations, where the likelihood function is non-convex in model parameters (for e.g., in non-linear SSMs with non-Gaussian noise), numerical optimization routines either yield locally optimal (or biased) estimates \cite{A2004,C2009} or require careful tuning of the algorithm parameters \cite{Kantas2009}. Finally, high computational cost of  ML based algorithms (compared to  Bayesian estimators) coupled with applicability to a restricted non-linear model class, renders ML based methods unsuitable for processes, that require fast on-line estimators. Bayesian methods, on the contrary are `optimization-free' estimators, which allow these methods, to be fast, and free from issues related to optimization. Comparisons between the ML and Bayesian based methods for parameter estimation are further drawn in Section \ref{sec:Chap1S7}. A summary of different Bayesian and ML based algorithms, including their advantages and disadvantages is presented in Table \ref{T0}.

In the next section, the motivation and the contributions of this paper are provided.
\section{Motivation and contributions}
The existing literature on Bayesian and likelihood based methods for on-line state-parameter estimation assumes that measurement will be available at all sampling time; however, in practice, missing measurements are common in the process industries, where measurements may not arrive or be available at all sampling time instants. The importance of developing algorithms under missing measurements is well recognized \cite{G1995}. Existing literature addresses the issues related to missing data in linear \cite{SS2000} and non-linear \cite{G2008} systems only under an off-line setting. Unfortunately, these methods cannot handle missing data in real-time.

In this paper, a complete approach to on-line Bayesian state and parameter estimation in non-linear SSMs with non-Gaussian noise is developed, using an extended state vector representation with artificial dynamics for the parameters. Since this approach treats the simultaneous state and parameter estimation problems as the same, it will simply be referred to as an estimation problem unless otherwise warranted. Due to the inherent limitations of the EKF and UKF based simultaneous state-parameter estimators, a particle based SIR filtering approach is used. The choice of the SIR filter is motivated by the fact that it is relatively (compared to ASIR filter) less sensitive to large process noise and is computationally less expensive. Furthermore, the importance weights are easily evaluated and the importance functions can be easily sampled \cite{R2004}. 

It is emphasized that the PFs can be made arbitrarily accurate by simply increasing the number of particles; however, this comes at a computational cost. Several authors have focussed on this issue and developed methods, which either allows adaptation of the particle sample size \cite{SS2006,FL2007} or the adaptation of the proposal distribution from which the particles are sampled \cite{DGA2000,F2008}. Performance of PFs is closely related to the ability to sample particles in state-space regions, where the posterior is significant \cite{P1999}. Perfect adaptation of the particle size or choice of an efficient proposal density for PFs is a long-standing topic (see \cite{CMO2008} for recent developments in this area). 

The following are the main contributions of this paper: (a) an adaptive SIR (Ad-SIR) filter for on-line state-parameter estimation in general non-linear SSMs with non-Gaussian noise is proposed and derived; (b) an optimal tuning rule to control the width of the kernel, and the variance of the artificial noise is proposed; (c) an on-line optimization algorithm based on KL divergence is used to project importance samples around the region of high likelihood, which allows adaptation of the SIR filter for on-line state-parameter estimation; (d) an extension of the algorithm to handle missing measurements in real-time is also presented; and (e) the efficacy of the algorithm is illustrated through numerical examples. 

The proposed algorithm can estimate states and parameters of both time-invariant and slowly time-variant stochastic non-linear systems. It exhibits good performance even for systems with large process or measurement noise.  A distinct advantage of the proposed algorithm is that it can also estimate parameters of the noise models. This particular feature is crucial, since filtering performance for any linear or non-linear filter depends on accurate characterization of the state and measurement noise models \cite{B2011}.
\section{Problem formulation}
\label{sec:Chap1S2}
Consider the following class of discrete-time, stochastic non-linear SSMs:
 \begin{subequations}
 \label{eq:Chap1E1}
 \begin{align}
     {X}_{t+1} ={f}_t({X}_{t},{u}_{t},{\theta}_t,V_t) , \label{eq:Chap1E1a}\\
    {Y}_t ={g}_t({X}_{t},{u}_{t},{\theta}_t,W_t),\label{eq:Chap1E1b}
 \end{align}
 \end{subequations}
where ${X}_t \in {\mathcal{X}} \subseteq \mathbb{R}^{n}$ and ${Y}_t \in {\mathcal{Y}} \subseteq \mathbb{R}^{m}$ for $t\in\mathbb{N}$ are the state and measurement processes, respectively. Here ${\mathbb{R}:=(-\infty, \infty)}$ and ${\mathbb{N}:=\{1,2,\dots,\}}$. $X_t\in {\mathcal{X}}$ is a Markov process, which is either partially or fully hidden, and ${Y_t\in {\mathcal{Y}}}$ may include missed measurements; ${{u}_{t} \in \mathcal{U}\subseteq\mathbb{R}^{p}}$ and ${{\theta}_t \in {\Theta} \subseteq \mathbb{R}^{r}}$ are the time-varying or time-invariant control variables and model parameters, respectively. The process and measurement noise are represented as ${V_t\in\mathbb{R}^n}$ and ${W_t \in\mathbb{R}^m}$, respectively. ${f_t(\cdot)}$ is a $n$-dimensional state mapping function and ${g_t(\cdot)}$ is a $m$-dimensional output mapping function, each being non-linear in its arguments, and possibly time-varying, such that ${{f}_t:=\mathcal{X}\times \mathcal{U}\times {\Theta}\times\mathbb{R}^n  \rightarrow \mathcal{X}}$ and ${{g}_t:=\mathcal{X}\times \mathcal{U}\times {\Theta}\times\mathbb{R}^m \rightarrow \mathcal{Y}}$. The assumption on (\ref{eq:Chap1E1}) is discussed next.
\begin{assumption}
\label{assumption:Chap1A1}
${V}_t\in\mathbb{R}^n$ and ${W}_t\in\mathbb{R}^m$ are the mutually independent sequences of independent random variables described by the probability density functions (pdfs)  ${p}(v_t|\cdot)$ and $p(w_t|\cdot)$, respectively. The pdfs  are known {a priori} in their classes ({e.g.}, Gaussian; Binomial)  and are parametrized by a finite number of moments ({e.g.}, mean; variance). If the moments are unknown, it can be augmented with the model parameter set ${\theta}_t \in {\Theta}$.
\end{assumption}
Since ${\theta}_t \in {\Theta}$ does not have an explicit transition function like $f_t(\cdot)$ for $X_t\in\mathcal{X}$, artificial dynamics are introduced, such that ${\theta}_t \in {\Theta}$ evolves according to
\begin{align}
\label{eq:Chap1E2}
{\theta}_{t+1}={\theta}_{t}+{\xi}_{t},
\end{align}
where ${\xi}_{t}\in\mathbb{R}^r$ is a sequence of independent Gaussian random variables realized from $\mathcal{N}(\xi_t|0, \Sigma_{\theta_t})$, independent of the noise sequences ${V}_t\in\mathbb{R}^n$ and ${W}_t\in\mathbb{R}^m$. The dynamics of $\theta_t$ in  (\ref{eq:Chap1E2}) is governed by the artificial noise variance ${\Sigma_{\theta_t}\in\mathcal{S}_+^r}$, where $\mathcal{S}_+^r$ is a cone of positive semi-definite matrix. Often $\Sigma_{\theta_t}$ is unknown, and requires careful tuning. The formulation in (\ref{eq:Chap1E2}) is the ADA, which avoids the parameter degeneracy problem discussed in Section \ref{sec:Chap1S1}, and further allows for estimation of time-varying parameters. 

Equations (\ref{eq:Chap1E1}) and (\ref{eq:Chap1E2}) together represent an extended SSM. For notational simplicity, the extended state vector is defined as ${Z_t\triangleq\{X_t, \theta_t\}}$, such that ${Z_t\in \mathcal{Z}\subseteq\mathbb{R}^{s=n+r}}$. Throughout this paper ${Z_t\in \mathcal{Z}}$ will be considered; however, distinction between the states and parameters will be made, as required. Equations (\ref{eq:Chap1E1}) and (\ref{eq:Chap1E2}) can be represented as:
\begin{subequations}
 \label{eq:Chap1E3}
 \begin{align}
 X_0\sim&p(x_0);\quad{X}_{t+1}|Z_t\sim p({x}_{t+1}|z_{t}); \label{eq:Chap1E3a}\\
 \theta_0\sim&p(\theta_0);~~\quad {\theta}_{t+1}|\theta_t\sim p({\theta}_{t+1}|\theta_{t}); \label{eq:Chap1E3b}\\
&\quad\quad~~~~~~~~~~~{Y}_t|Z_t \sim p({y}_{t}|{z}_{t}),\label{eq:Chap1E3c}
 \end{align}
 \end{subequations}
where: the Markov process ${{X}_t\in\mathcal{X}}$ is characterized by its initial density $p(x_0)$ and a transition density $p(x_{t+1}|z_t)$, while the Markov process ${{\theta}_t\in\Theta}$ is characterized by its initial density $p(\theta_0)$ and a transition density $p(\theta_{t+1}|\theta_t)$. The measurement ${Y_t\in\mathcal{Y}}$ is assumed to be conditionally independent given ${Z_t\in\mathcal{Z}}$, and is characterized by the conditional marginal density $p({y}_{t}|{z}_{t})$.  The representation in (\ref{eq:Chap1E3}) includes a wide class of non-linear time-series models, including (\ref{eq:Chap1E1}). For the sake of clarity, the input signal ${u}_{t}\in\mathcal{U}$ is omitted in (\ref{eq:Chap1E3}); however, all the derivations that appear in this paper hold with ${u}_{t}\in\mathcal{U}$ included. 

The main problems addressed in this paper are stated next.
\begin{problem}
\label{problem:Chap1P1}
The first problem aims at computing the state-parameter estimate of ${{Z_t\in\mathcal{Z}}}$ in real-time using $\{{u}_{1:t};{y}_{1:t}\}$; wherein, ${y}_{1:t}\triangleq\{{y}_1,\dots,{y}_t \}$  is a vector of measured outputs corresponding to the input sequence ${u}_{1:t}\triangleq\{{u}_1,\dots,{u}_t \}$.
\end{problem}
\begin{problem}
\label{problem:Chap1P2}
The second problem aims at computing the state-parameter estimate of ${Z_t\in\mathcal{Z}}$ in real-time using $\{{u}_{1:t};{y}_{t_1:t_\gamma}\}$; wherein, the measurements arrive at random sampling time instants, such that only $\{y_{t_1},\dots,y_{t_\gamma}\}$ out of $y_{1:t}$ is available.
\end{problem}
\section{Bayesian filtering}
\label{sec:Chap1Extra}
The Bayesian idea for solving Problems \ref{problem:Chap1P1} and \ref{problem:Chap1P2} is to construct a posterior pdf ${Z_t|(Y_{1:t}=y_{1:t})\sim p({z}_t| {y}_{1:t})}$ for all ${t\in\mathbb{N}}$. Here $p({z}_t| {y}_{1:t})$ is a probabilistic representation of available statistical information on ${Z_t\in\mathcal{Z}}$ conditioned on ${\{Y_{1:t}={y}_{1:t}\}}$. Using the Markov property of (\ref{eq:Chap1E3}) and from the Bayes' theorem, $p({z}_t| {y}_{1:t})$ can be computed as
\begin{align}
\label{eq:Chap1E8}
p({z}_t| {y}_{1:t})&=\frac{p({y}_t |{z}_t)p({z}_t| {y}_{1:t-1})}{p({y}_t | {y}_{1:t-1})},
\end{align}
where: $p({y}_t | {y}_{1:t-1})=\int_{\mathcal{Z}} p({y}_t|{z}_t)p(dz_t|{y}_{1:t-1})$ is a constant; $p(dz_t|{y}_{1:t-1})\triangleq p(z_t|{y}_{1:t-1})d{z}_t$ is a prior distribution; and $p({z}_t | {y}_{1:t-1})$ is a prior density, which can be computed as
\begin{align}
\label{eq:Chap1E7}
p({z}_t | {y}_{1:t-1})=&\int_{\mathcal{Z}}p({z}_t| {z}_{t-1})p(d{z}_{t-1}| {y}_{1:t-1}),
\end{align}
where $p(d{z}_{t-1} | {y}_{1:t-1})\triangleq p({z}_{t-1}| {y}_{1:t-1})d{z}_{t-1}$ is the posterior distribution at ${t-1}$. Ignoring the constant term,  (\ref{eq:Chap1E8}) in compact form can be written as follows
\begin{align}
\label{eq:Chap1E9}
p(z_t  | {y}_{1:t})&\propto {p({y}_t |z_t )}p(z_t | {y}_{1:t-1}).	
\end{align}
In principle, the recurrence relation between the prediction and update equations in (\ref{eq:Chap1E7}) and (\ref{eq:Chap1E9}), respectively, provides a complete Bayesian solution to Problems \ref{problem:Chap1P1} and \ref{problem:Chap1P2}. 

To compute a point estimate from $p(z_t|y_{1:t})$, a common approach is to minimize the mean-square error (MSE) risk ${\mathcal{R}_Z\triangleq\mathbb{E}_{p(Z_t,Y_{1:t})}[\|Z_t-\widehat{Z}_{t|t}\|^2_2]}$, where ${\widehat{Z}_{t|t}\in\mathbb{R}^{s}}$ is the point estimate of the states and parameters at time ${t\in\mathbb{N}}$; ${\|\cdot\|_2}$ is a $2-$norm operator; and $\mathbb{E}_{p(\cdot)}$ is the expectation with respect to the pdf $p(\cdot)$. Minimizing $\mathcal{R}_Z$ over $\widehat{Z}_{t|t}$ yields conditional mean of ${Z_t|(Y_{1:t}=y_{1:t})\sim p({z}_t| {y}_{1:t})}$ as an optimal point estimate \cite{HL1968}. For instance, if  $\mathcal{R}_\theta$ is the MSE Bayes' risk then the MMSE parameter estimate is given by
\begin{align}
\label{eq:Chap1E6}
\widehat{\theta}_{t|t}\triangleq \mathbb{E}_{p(\theta_t|Y_{1:t})}[{\theta}_t]=\int_{\Theta}{\theta}_t p(d{\theta}_t| {y}_{1:t}),
\end{align}
where $p(d{\theta_t}| {y}_{1:t})$ is the marginalized posterior distribution for the parameters, such that
\begin{equation}
\label{eq:Chap1E5}
p(d{\theta_t}| {y}_{1:t}) =\int_{\mathcal{X}}{p(dz_t|y_{1:t})}.
\end{equation}
\begin{rmk}
Except for linear systems with Gaussian state and measurement noise or when $\mathcal{Z}$ is a finite set, with finite computing capabilities, Bayesian on-line state-parameter estimation solution given in (\ref{eq:Chap1E9}) cannot be solved exactly. 
\end{rmk}
This paper proposes an SMC based adaptive SIR filter to numerically approximate the Bayesian on-line state-parameter estimation solution given in (\ref{eq:Chap1E9}).
\section{Adaptive SIR filter}
\label{sec:Chap1S4}
It is not our aim to review SMC methods in details, but simply to point out their intrinsic limitations, which have fundamental practical consequences on the ADA introduced in Section \ref{sec:Chap1S2}. The essential idea behind SMC methods is to generate a set of random particles and their associated weights from the target pdf. The target pdf of interest here is the posterior pdf $p(z_t|y_{1:t})$ in (\ref{eq:Chap1E9}). Unfortunately, due to the non-Gaussian nature of $p(z_t|y_{1:t})$, generating set of random particles from the target pdf is non-trivial \cite{R2004}. 

An alternate idea is to employ importance sampling function (ISF) $q(z_t|y_{1:t},z_{t-1})$, such that $q(z_t|y_{1:t},z_{t-1})$ is a non-negative function on $\mathcal{Z}$ and $\supp q(z_t|y_{1:t},z_{t-1})\supseteq \supp p(z_t|y_{1:t})$. A standard SIR filter selects ${q(z_t|y_{1:t},z_{t-1})=p({z}_t| {y}_{1:t-1})}$ \cite{A2002}, since it enables easy sampling from the ISF and easy evaluation of $p({z}_t| {y}_{1:t-1})$ for any ${\{Z_t,~Y_{1:t-1}\}\in\mathcal{Z}\times\mathcal{Y}^{t-1}}$. Now to generate a set of random particles from the ISF $p({z}_t|{y}_{1:t-1})$, the multi-dimensional integral in (\ref{eq:Chap1E7}) needs to be evaluated first.
Using samples from $p({z}_{t-1}|{y}_{1:t-1})$ (available from the recursive relation in (\ref{eq:Chap1E7}) and (\ref{eq:Chap1E9})), an SMC approximation of the posterior distribution ${{Z}_{t-1}|(Y_{1:t-1}={y}_{1:t-1})\sim p(dz_{t-1}|y_{1:t-1})}$ is given by
\begin{align}
\label{eq:Chap1E12}
\tilde{p}(d{z}_{t-1}|{y}_{1:t-1})=& \sum_{i=1}^{N} W^i_{t-1|t-1}\delta_{{Z}^i_{t-1|t-1} }(d{z}_{t-1}),
\end{align}
where: ${\tilde{p}(d{z}_{t-1}|{y}_{1:t-1})}$ is an SMC estimate of the joint state-parameter posterior distribution ${p(d{z}_{t-1}|{y}_{1:t-1})}$; ${\{Z^i_{t-1|t-1};~W^i_{t-1|t-1}\}_{i=1}^N\sim\tilde{p}(z_{t-1}|{y}_{1:t-1})}$ is a set of $N$ particles and their weights, distributed according to ${\tilde{p}({z}_{t-1}|{y}_{1:t-1})}$, such that ${\sum_{i=1}^{N} W^i_{t-1|t-1}=1}$ and ${\delta_{Z^i_{t-1|t-1}}(dz_{t-1})}$ is the Dirac delta mass located at the random sample ${Z^i_{t-1|t-1}}$.  

Using (\ref{eq:Chap1E12}), an SMC approximation of the marginalized posterior distribution of the states and parameters  at ${t-1}$ can also be computed as given in the next lemma.
\begin{lemma}
\label{lemma:Chap1L1}
Let the SMC approximation of the distribution of ${{Z}_{t-1}|(Y_{1:t-1}={y}_{1:t-1})}$ be given by (\ref{eq:Chap1E12}) then marginalizing (\ref{eq:Chap1E12}) over ${{X}_t \in\mathcal{X}}$ and ${{\theta}_t \in{\Theta}}$ yields approximate distributions for ${{\theta}_{t-1}|(Y_{1:t-1}={y}_{1:t-1})}$ and ${{X}_{t-1}|(Y_{1:t-1}={y}_{1:t-1})}$, respectively, such that
\begin{subequations}
\label{eq:Chap1EL}
\begin{align}
\tilde{p}(d{\theta}_{t-1} | {y}_{1:t-1})=\sum_{i=1}^{N} W^i_{t-1|t-1}\delta_{{\theta}^i_{t-1|t-1}}({d\theta}_{t-1} ),\label{eq:Chap1EL1}\\
\tilde{p}(d{x}_{t-1} | {y}_{1:t-1})=\sum_{i=1}^{N} W^i_{t-1|t-1}\delta_{{X}^i_{t-1|t-1}}(d{x}_{t-1} )\label{eq:Chap1EC1},
\end{align}
\end{subequations}
where $\tilde{p}(d{\theta}_{t-1} | {y}_{1:t-1})$ and $\tilde{p}(d{x}_{t-1} | {y}_{1:t-1})$ are the SMC approximations of the distributions $p(d{\theta}_{t-1} | {y}_{1:t-1})$ and $p(d{x}_{t-1} | {y}_{1:t-1})$, respectively.
\end{lemma}
\begin{proof}
Using the Law of Total Probability on posterior distribution $p(d{z}_{t-1}|{y}_{1:t-1})$ yields
\begin{align}
\label{eq:Chap1E13}
p({d\theta}_{t-1} |y_{1:t-1})&=\int_\mathcal{X} p(d{z}_{t-1}| {y}_{1:t-1}).
\end{align}
Substituting (\ref{eq:Chap1E12}) into (\ref{eq:Chap1E13}) and taking independent terms outside the integral yields
\begin{subequations}
\begin{align}
\tilde{p}(d{\theta}_{t-1} | {y}_{1:t-1})=&\sum_{i=1}^{N} W^i_{t-1|t-1}\int_\mathcal{X}\delta_{{Z}^i_{t-1|t-1}}(d{z}_{t-1}),\\
=&\sum_{i=1}^{N} W^i_{t-1|t-1}\delta_{{\theta}^i_{t-1|t-1}}(d{\theta}_{t-1}).\label{eq:Chap1E14}
\end{align}
\end{subequations}
The equality in (\ref{eq:Chap1E14}) is a result from marginalization of the joint state-parameter Dirac delta function over $\mathcal{X}$, which completes the proof. 
\end{proof}
Lemma \ref{lemma:Chap1L1} computes the marginal distributions of ${{\theta}_{t-1}|(Y_{1:t-1}={y}_{1:t-1})}$ and ${{X}_{t-1}|(Y_{1:t-1}={y}_{1:t-1})}$ using (\ref{eq:Chap1E12}). Note that the weights in (\ref{eq:Chap1EL}) are same as that in (\ref{eq:Chap1E12}).
\begin{rmk}
\label{remark:Chap1R2}
From (\ref{eq:Chap1EL1}), the mean and the covariance of ${\theta_{t-1}|(Y_{1:t-1}={y}_{1:t-1})}$ can be approximated as ${\mathbb{E}_{p(\theta_{t-1}|Y_{1:t-1})}\left[{\theta}_{t-1}\right]\triangleq\int_{\Theta}\theta_{t-1}p(d\theta_{t-1}|y_{1:t-1})\approx\sum_{i=1}^{N} W^i_{t-1|t-1}{\theta}^i_{t-1|t-1}}$ $ {=\widehat{{\theta}}_{t-1|t-1}}$ and ${\mathbb{V}_{p(\theta_{t-1}|Y_{1:t-1})}\left[{\theta}_{t-1}\right]\triangleq}$ $\int_{\Theta}(\theta_{t-1}-\widehat{\theta}_{t-1|t-1})(\theta_{t-1}-\widehat{\theta}_{t-1|t-1})^T p(d\theta_{t-1}|y_{1:t-1})$ $\approx\sum_{i=1}^{N} W^i_{t-1|t-1}({\theta}^i_{t-1|t-1}-\widehat{\theta}_{t-1|t-1})({\theta}^i_{t-1|t-1}-\widehat{\theta}_{t-1|t-1})^T$ $=V_{\theta_{t-1}}$, respectively.
\end{rmk}
In Remark \ref{remark:Chap1R2}, ${\widehat{\theta}_{t-1|t-1}\in\mathbb{R}^r}$ is an MMSE parameter estimate at $t-1$. Similarly, an MMSE state estimate ${\widehat{X}_{t-1|t-1}\in\mathbb{R}^n}$ at $t-1$ can also be computed using (\ref{eq:Chap1EC1}). Finally, to generate a set of random particles from the ISF, substituting (\ref{eq:Chap1E12}) into (\ref{eq:Chap1E7}) yields
\begin{subequations}
\begin{align}
\tilde{p}(z_t | {y}_{1:t-1})=& \int_{\mathcal{Z}}p({z}_t |{z}_{t-1})\sum_{i=1}^{N} W^i_{t-1|t-1}\delta_{{Z}^i_{t-1|t-1}}(d{z}_{t-1} ),\\
=&\sum_{i=1}^{N} W^i_{t-1|t-1}p({z}_t | {Z}^i_{t-1|t-1}),\label{eq:Chap1E16}
\end{align}
\end{subequations}
where $\tilde{p}(z_t | {y}_{1:t-1})$ is an SMC approximation of the ISF $p(z_t | {y}_{1:t-1})$. The approximation in (\ref{eq:Chap1E16}) is a mixture of $N$ transitional pdfs, with a mixing ratio $\{W^i_{t-1|t-1}\}_{i=1}^N$ and centred at $\{{Z}^i_{t-1|t-1}\}_{i=1}^N$. Marginalization of the ISF $p(z_t | {y}_{1:t-1})$ over ${X}_t \in{\mathcal{X}}$ is discussed in next.
\begin{lemma}
\label{lemma:Chap1L2}
Let ${\xi_t\in\mathbb{R}^r}$ in (\ref{eq:Chap1E3b}) be a sequence of independent Gaussian variable, such that ${\xi_t\sim\mathcal{N}(\xi_t|0,\Sigma_{\theta_t})}$, where ${\Sigma_{\theta_t}\in\mathcal{S}_+^r}$ for all ${t\in\mathbb{N}}$ then marginalizing (\ref{eq:Chap1E16}) over ${X_t \in\mathcal{X}}$ yields a mixture Gaussian pdf for ${\theta_t |(Y_{1:t-1}=y_{1:t-1})}$ given by
\begin{align}
\label{eq:Chap1EE16}
\tilde{p}({\theta}_t | {y}_{1:t-1})=\sum_{i=1}^{N}W^i_{t-1|t-1}\mathcal{N}(\theta_t|{\theta}^i_{t-1|t-1}, \Sigma_{\theta_t}),
\end{align}
where ${\theta_t|\theta^i_{t-1|t-1}\sim \mathcal{N}(\theta_t|\theta^i_{t-1|t-1},\Sigma_{\theta_t})}$ follws a Gaussian density with mean ${\theta^i_{t-1|t-1}\in\mathbb{R}^r}$ and covariance  ${\Sigma_{\theta_t}\in\mathcal{S}_+^r}$.
\end{lemma}
\begin{proof}
Using the Law of Total Probability on the ISF $p({z}_{t} | {y}_{1:t-1})$ yields
\begin{align}
\label{eq:Chap1E17}
p({\theta}_t | {y}_{1:t-1})=\int_\mathcal{X} p({z}_{t}| {y}_{1:t-1})d{x}_t.
\end{align}
Substituting (\ref{eq:Chap1E16}) into (\ref{eq:Chap1E17}) and pulling independent terms out of the integral yields
\begin{subequations}
\begin{align}
\tilde{p}({\theta}_t | {y}_{1:t-1})=&\sum_{i=1}^{N} W^i_{t-1|t-1}\int_\mathcal{X}p({z}_t | {Z}^i_{t-1|t-1})dx_t,\\
=&\sum_{i=1}^{N}W^i_{t-1|t-1}p({\theta}_t|{\theta}^i_{t-1|t-1})\int_\mathcal{X} p({x}_t |{Z}^i_{t-1|t-1})d{x}_t,\label{eq:Chap1E18}
\end{align}
\end{subequations}
where ${\tilde{p}({\theta}_t | {y}_{1:t-1})}$ is an estimate. Since, $\int_\mathcal{X} p(d{x}_t | {Z}^i_{t-1|t-1})=1$, (\ref{eq:Chap1E18}) simplifies to
\begin{subequations}
\begin{align}
\tilde{p}({\theta}_t | {y}_{1:t-1})&=\sum_{i=1}^{N}W^i_{t-1|t-1}p({\theta}_t|{\theta}^i_{t-1|t-1}),\\
&=\sum_{i=1}^{N}W^i_{t-1|t-1}\mathcal{N}(\theta_t|{\theta}^i_{t-1|t-1}, \Sigma_{\theta_t}).\label{eq:Chap1E19}
\end{align}
\end{subequations}
The equality in (\ref{eq:Chap1E19}) follows from the fact that the pdf $p({\theta}_t|{\theta}^i_{t-1|t-1})$ models the noise distribution ${\xi}_{t}\sim\mathcal{N}(\xi_t|0,\Sigma_{\theta_t})$ (see (\ref{eq:Chap1E3b})). 
\end{proof}
\cite{LW2001,W1993} refer to (\ref{eq:Chap1EE16}) as Gaussian kernel estimate of the marginalized ISF, whose kernel width is controlled by the noise covariance $\Sigma_{\theta_t}$.  Statistics of (\ref{eq:Chap1EE16}) are given next to highlight the implications of using SMC methods with ADA.
\begin{lemma}
\label{lemma:Chap1L3}
Let the artificial noise in (\ref{eq:Chap1E3b}) be ${\xi_t\sim\mathcal{N}(\xi_t|0,\Sigma_{\theta_t})}$ and let ${\widehat{\theta}_{t-1|t-1}\in\mathbb{R}^r}$ and ${V_{\theta_{t-1}}\in\mathcal{S}_+^r}$ be the mean and covariance of ${\theta_{t-1}|(Y_{1:t-1}=y_{1:t-1})\sim\tilde{p}(\theta_{t-1}|y_{1:t-1})}$ as computed in Remark \ref{remark:Chap1R2}. Also, let the SMC approximation of the marginalized ISF be given by (\ref{eq:Chap1EE16}), such that ${\theta_t|(Y_{1:t-1}=y_{1:t-1})\sim\tilde{p}(\theta_t|y_{1:t-1})}$ then the first and second moment of ${\theta_t|(Y_{1:t-1}=y_{1:t-1})}$ is given by
\begin{subequations}
\begin{align}
\mathbb{E}_{p(\theta_t|Y_{1:t-1})}[\theta_t]&=\widehat{\theta}_{t-1|t-1},\label{eq:Chap1EE17}\\
\mathbb{V}_{p(\theta_t|Y_{1:t-1})}[\theta_t]&= V_{\theta_{t-1}}+\Sigma_{\theta_t} \label{eq:Chap1EE18}.
\end{align}
\end{subequations}
\end{lemma}
\begin{proof}
Expectation of ${\theta_t|(Y_{1:t-1}=y_{1:t-1})}$ is given by
\begin{align}
\label{eq:Chap1E20}
\mathbb{E}_{p(\theta_t|Y_{1:t-1})}[\theta_t]=\int_{\Theta}\theta_t p(d\theta_t|y_{1:t-1}).
\end{align}
Substituting (\ref{eq:Chap1E19}) into (\ref{eq:Chap1E20}) yields
\begin{subequations}
\begin{align}
\mathbb{E}_{p(\theta_t|Y_{1:t-1})}[\theta_t]=&\int_{\Theta}\theta_t \sum_{i=1}^{N} W_{t-1|t-1}^i\mathcal{N}(d\theta_t|\theta^i_{t-1|t-1},\Sigma_{\theta_t}),\\
=&\sum_{i=1}^{N} W_{t-1|t-1}^i\int_{\Theta}\theta_t \mathcal{N}(d\theta_t|\theta^i_{t-1|t-1},\Sigma_{\theta_t}),\\
=&\sum_{i=1}^{N} W_{t-1|t-1}^i\theta^i_{t-1|t-1}=\widehat{\theta}_{t-1|t-1},\label{eq:Chap1E21}
\end{align}
\end{subequations}
where (\ref{eq:Chap1E21}) is from Remark \ref{remark:Chap1R2}, which completes the proof for (\ref{eq:Chap1EE17}). Now the covariance of ${\theta_t|(Y_{1:t-1}=y_{1:t-1})}$ is given by
\begin{align}
\label{eq:Chap1E22}
\mathbb{V}_{p(\theta_t|Y_{1:t-1})}[\theta_t]=\int_{\Theta}(\theta_t-\mathbb{E}_{p(\theta_t|Y_{1:t-1})}[\theta_t])(\theta_t-\mathbb{E}_{p(\theta_t|Y_{1:t-1})}[\theta_t])^T p(d\theta_t|y_{1:t-1}).
\end{align}
Substituting (\ref{eq:Chap1E19}) and (\ref{eq:Chap1E21}) into (\ref{eq:Chap1E22}) yields
\begin{align}
\label{eq:Chap1E23}
\mathbb{V}_{p(\theta_t|Y_{1:t-1})}[\theta_t]=\sum_{i=1}^{N}W^i_{t-1|t-1}\int_{\Theta}(\theta_t-\widehat{\theta}_{t-1|t-1})(\theta_t-\widehat{\theta}_{t-1|t-1})^T\mathcal{N}(d\theta_t|\theta^i_{t-1|t-1},\Sigma_{\theta_t}).
\end{align}
Simple algebraic manipulation of (\ref{eq:Chap1E23}) yields
\begin{align}
\label{eq:Chap1E24}
\mathbb{V}_{p(\theta_t|Y_{1:t-1})}[\theta_t]= &\sum_{i=1}^{N}W^i_{t-1|t-1}\int_{\Theta}(\theta_t-\theta^{i}_{t-1|t-1}+\theta^{i}_{t-1|t-1}-\widehat{\theta}_{t-1|t-1})\nonumber\\
&\times(\theta_t-\theta^{i}_{t-1|t-1}+\theta^{i}_{t-1|t-1}-\widehat{\theta}_{t-1|t-1})^T\mathcal{N}(d\theta_t|\theta^i_{t-1|t-1},\Sigma_{\theta_t}).
\end{align}
Simplifying the terms in (\ref{eq:Chap1E24}) and representing the integral solution as
\begin{align}
\label{eq:Chap1E25}
\mathbb{V}_{p(\theta_t|Y_{1:t-1})}[\theta_t]= &I_1+I_2+I_3+I_4,
\end{align}
where:
\begin{subequations}
\begin{flalign}
&I_1=\sum_{i=1}^{N}W^i_{t-1|t-1}\int_{\Theta}(\theta_t-\theta^{i}_{t-1|t-1})(\theta_t-\theta^{i}_{t-1|t-1})^T\mathcal{N}(d\theta_t|\theta^i_{t-1|t-1},\Sigma_{\theta_t})=\sum_{i=1}^{N}W^i_{t-1|t-1}\Sigma_{\theta_t}&\nonumber\\
&=\Sigma_{\theta_t};&\label{eq:Chap1E26}
\end{flalign}
\begin{flalign}
&I_2=\sum_{i=1}^{N}W^i_{t-1|t-1}\int_{\Theta}(\theta^i_{t-1|t-1}-\widehat{\theta}_{t-1|t-1})(\theta^i_{t-1|t-1}-\widehat{\theta}_{t-1|t-1})^T\mathcal{N}(d\theta_t|\theta^i_{t-1|t-1},\Sigma_{\theta_t})=&\nonumber\\
&\sum_{i=1}^{N}W^i_{t-1|t-1}(\theta^i_{t-1|t-1}-\widehat{\theta}_{t-1|t-1})(\theta^i_{t-1|t-1}-\widehat{\theta}_{t-1|t-1})^T\int_{\Theta}\mathcal{N}(d\theta_t|\theta^i_{t-1|t-1},\Sigma_{\theta_t})=V_{\theta_{t-1}};&\label{eq:Chap1E27}
\end{flalign}
\begin{flalign}
&I_3=\sum_{i=1}^{N}W^i_{t-1|t-1}\int_{\Theta}(\theta_t-\theta^{i}_{t-1|t-1})(\theta^i_{t-1|t-1}-\widehat{\theta}_{t-1|t-1})^T\mathcal{N}(d\theta_t|\theta^i_{t-1|t-1},\Sigma_{\theta_t})=&\nonumber\\
&\sum_{i=1}^{N}W^i_{t-1|t-1}\int_{\Theta}(\theta_t-\theta^{i}_{t-1|t-1}) \mathcal{N}(d\theta_t|\theta^i_{t-1|t-1},\Sigma_{\theta_t})(\theta^i_{t-1|t-1}-\widehat{\theta}_{t-1|t-1})^T=0;&\label{eq:Chap1E28}
\end{flalign}
\begin{flalign}
&I_4=\sum_{i=1}^{N}W^i_{t-1|t-1}\int_{\Theta}(\theta^i_{t-1|t-1}-\widehat{\theta}_{t-1|t-1})(\theta_t-\theta^{i}_{t-1|t-1})^T\mathcal{N}(d\theta_t|\theta^i_{t-1|t-1},\Sigma_{\theta_t})=&\nonumber\\
&\sum_{i=1}^{N}W^i_{t-1|t-1}(\theta^i_{t-1|t-1}-\widehat{\theta}_{t-1|t-1})\int_{\Theta}(\theta_t-\theta^{i}_{t-1|t-1})^T\mathcal{N}(d\theta_t|\theta^i_{t-1|t-1},\Sigma_{\theta_t})=0.&\label{eq:Chap1E29}
\end{flalign}
\end{subequations}
Here (\ref{eq:Chap1E26}) and (\ref{eq:Chap1E27}) are based on Remark \ref{remark:Chap1R2}, and (\ref{eq:Chap1E28}) and (\ref{eq:Chap1E29}) use the relation $\int_{\Theta}\theta_t\mathcal{N}(\theta_t|\theta^i_{t-1|t-1},\Sigma_{\theta_t}) d\theta_t=\theta^i_{t-1|t-1}$. Finally, substituting (\ref{eq:Chap1E26}), (\ref{eq:Chap1E27}), (\ref{eq:Chap1E28}) and (\ref{eq:Chap1E29}) into (\ref{eq:Chap1E25}) yields (\ref{eq:Chap1EE18}), which completes the proof.
\end{proof}
\begin{rmk}
\label{remark:Chap1RE3}
From Remark \ref{remark:Chap1R2} and Lemma \ref{lemma:Chap1L3}, while computing $\tilde{p}(\theta_t|y_{1:t-1})$ from $\tilde{p}(\theta_{t-1}|y_{1:t-1})$, the mean is unchanged, i.e., $\mathbb{E}_{p(\theta_{t-1}|Y_{1:t-1})}[\theta_{t-1}]=\mathbb{E}_{p(\theta_{t}|Y_{1:t-1})}[\theta_{t}]$, while the covariance disperses by $\Sigma_{\theta_t}$, such that $\mathbb{V}_{p(\theta_{t}|Y_{1:t-1})}[\theta_{t}]-\mathbb{V}_{p(\theta_{t-1}|Y_{1:t-1})}[\theta_{t-1}]=\Sigma_{\theta_t}$.
\end{rmk}
Remark \ref{remark:Chap1RE3} highlights the variance inflation problem associated with the ADA. In \cite{LW2001}, the authors implied similar results. Note that the results presented here are important, since they are the key aspects underlying the Ad-SIR filter proposed here.
\subsection{Kernel smoothing}
\label{sec:Chap1S4.2}
It is well known that using particles sampled from an over-dispersed ISF will yield a poor approximation of the posterior pdf \cite{LW2001}. From Remark \ref{remark:Chap1RE3}, it is clear that the SMC approximation of the  marginalized ISF in (\ref{eq:Chap1EE16}) suffers from a similar dispersion problem. To overcome the issue of dispersion, use of a kernel method is proposed. The idea behind this approach is the shrinkage of the kernel width according to
\begin{align}
\label{eq:Chap1E31}
{\tilde{\theta}}^i_{t-1|t-1}&= \sqrt{1-h_{t}^2}~{\theta}^i_{t-1|t-1} +\big(1-\sqrt{1-h_{t}^2}\big)~\widehat{\theta}_{t-1|t-1},
\end{align}
where ${\{\tilde{\theta}^i_{t-1|t-1}\}_{i=1}^N}$ are the shrinkage locations and ${h_{t} \in[0,1]}$ is a kernel parameter. Therefore replacing ${\{{\theta}^i_{t-1|t-1}\}_{i=1}^N}$ with ${\{{\tilde{\theta}}^i_{t-1|t-1}\}_{i=1}^N}$ in (\ref{eq:Chap1EE16}) and setting ${\Sigma_{\theta_t}=h_{t}^2V_{\theta_{t-1}}}$, the SMC approximation of the marginalized ISF in (\ref{eq:Chap1EE16}) can now be represented as
\begin{align}
\label{eq:Chap1E32}
\tilde{p}({\theta}_t | {y}_{1:t-1})=&\sum_{i=1}^{N}W^i_{t-1|t-1}\mathcal{N}(\theta_t|{\tilde{\theta}}^i_{t-1|t-1}, h_{t}^2V_{\theta_{t-1}}).
\end{align}
Note that by setting ${\Sigma_{\theta_t}=h_{t}^2V_{\theta_{t-1}}}$, the kernel width ${\Sigma_{\theta_t}}$ becomes a non-linear function of the kernel parameter $h_t$. Tuning of $h_t$ is discussed in Section \ref{sec:Chap1S4.3}, but first the statistics of (\ref{eq:Chap1E32}) as a plausible SMC approximation of the marginalized ISF are discussed next.
\begin{corollary}
\label{corollary:Chap1C2}
Let the SMC approximation of $p(\theta_t|y_{1:t-1})$  with kernel smoothing be represented by (\ref{eq:Chap1E32}) then the first two moments of ${\theta_t|(Y_{1:t-1}=y_{1:t-1})\sim \tilde{p}(\theta_t|y_{1:t-1})}$ are given by ${\mathbb{E}_{p(\theta_t|(Y_{1:t-1})}[{\theta}_{t}]=\widehat{\theta}_{t-1|t-1}}$ and ${\mathbb{V}_{p(\theta_t|(Y_{1:t-1})}[{\theta}_{t}]=V_{\theta_{t-1}}}$, respectively.
\end{corollary}
\begin{proof}
The proof is based on using (\ref{eq:Chap1E31}) and setting $\Sigma_{\theta_t}=h_{t}^2V_{\theta_{t-1}}$ in Lemma \ref{lemma:Chap1L3}. 
\end{proof}

With kernel smoothing, the SMC approximations of $\theta_{t}|(Y_{1:t-1}=y_{1:t-1})\sim \tilde{p}(\theta_t|y_{1:t-1})$ and $\theta_{t-1}|(Y_{1:t-1}=y_{1:t-1})\sim \tilde{p}(\theta_t|y_{1:t-1})$ have the same first two moments (see Corollary \ref{corollary:Chap1C2}). Finally, defining $\tilde{Z}^i_{t-1|t-1}\triangleq\{X^i_{t-1|t-1};~\tilde{\theta}^i_{t-1|t-1}\}$, the SMC approximation of the ISF density in (\ref{eq:Chap1E16}) with kernel smoothing can be represented as
\begin{align}
\label{eq:Chap1E33}
\tilde{p}({z}_t | {y}_{1:t-1})=&\sum_{i=1}^{N} W^i_{t-1|t-1}p({z}_t | {\tilde{Z}}^i_{t-1|t-1}).
\end{align}
Note that the random particle set ${\{{Z}^i_{t|t-1};W^i_{t|t-1}\}_{i=1}^N\sim \tilde{p}({z}_t| {y}_{1:t-1})}$ from (\ref{eq:Chap1E33}) can be generated by passing ${\{{\tilde{Z}}^i_{t-1|t-1}\}_{i=1}^N}$ through the transition pdfs, such that
\begin{subequations}
\label{eq:Chap1E34}
\begin{align}
X^i_{t|t-1}\sim&~p(x_t|\tilde{Z}^i_{t-1|t-1}),\\
\theta^i_{t|t-1}\sim&~p(\theta_t|\tilde{\theta}^i_{t-1|t-1}),
\end{align}
\end{subequations}
where $1\leq i\leq N$. Using the generated random particle set ${\{{Z}^i_{t|t-1};W^i_{t|t-1}\}_{i=1}^N}$ from (\ref{eq:Chap1E33}), an SMC approximation of the ISF distribution $p(d{z}_t| {y}_{1:t-1})$ can be represented as
\begin{align}
\label{eq:Chap1E35}
\tilde{p}(d{z}_t| {y}_{1:t-1})&=\sum_{i=1}^{N} W^i_{t|t-1}\delta_{Z^i_{t|t-1}}(d{z}_{t}),
\end{align}
where ${\{W^i_{t|t-1}=W^i_{t-1|t-1}\}_{i=1}^N}$. Now to obtain an SMC approximation of the target posterior distribution $p(d{z}_t| {y}_{1:t})$, substituting (\ref{eq:Chap1E35}) into (\ref{eq:Chap1E9}) yields
\begin{subequations}
\begin{align}
\tilde{p}(d{z}_t| {y}_{1:t})\propto &{p({y}_t |{z}_t)} \sum_{i=1}^{N}W^i_{t|t-1}\delta_{{Z}^i_{t|t-1}}(d{z}_{t}),\\	
=& \sum_{i=1}^{N} W^i_{t|t}\delta_{{Z}^i_{t|t-1}}(d{z}_{t}),\label{eq:Chap1E36}
\end{align}
\end{subequations}
where the weight ${W^i_{t|t}}$ in (\ref{eq:Chap1E36}) is given by
\begin{align}
\label{eq:Chap1E37}
W^i_{t|t}=  \frac{W^i_{t|t-1}p({y}_t | {Z}^i_{t|t-1})}{\sum_{i=1}^{N} {W^i_{t|t-1}p({y}_t |{Z}^i_{t|t-1})}}.
\end{align}
Note that in (\ref{eq:Chap1E36}) the importance weights ${\{W^i_{t|t}\}_{i=1}^N}$ are computed using the likelihood function. Finally, the MMSE point estimates for the states and parameters  at ${t\in\mathbb{N}}$ can be computed from (\ref{eq:Chap1E36}) using the procedure outlined in Lemma \ref{lemma:Chap1L1} and Remark \ref{remark:Chap1R2}.
\subsection{Optimal tuning of kernel parameter}
\label{sec:Chap1S4.3}
Although over-dispersion in the SMC approximation of the ISF is corrected using the kernel smoothing, optimal tuning of the kernel parameter ${h_t\in[0,1]}$ remains unclear.
\begin{rmk}
The tuning practices for ${h_t}$ are largely {ad-hoc}. \cite{LW2001} suggested selecting ${h_t=0.1}$; whereas, in \cite{C2005}, $h_t$ was optimized based on historical data-set, and then applied to future batches. These ad-hoc rules deliver a constant ${h_t}$, for which, optimality cannot be established with respect to the incoming data.
\end{rmk}
An optimal tuning rule for ${h_t}$ based on an on-line optimization procedure is proposed in this paper. The tuning rule is based on minimization of the KL divergence between the ISF and the target posterior density at each sampling time. The objective of the optimizer is not only to tune ${h_t}$, but to also project the particles sampled from the ISF in the region of high posterior density. This is to allow for adaptation of the SIR filter for combined state-parameter estimation. A similar idea of adaptive filtering is also proposed in \cite{CMO2008}. In a standard SIR filter, if ${\supp{p(z_t|y_{1:t-1})}}$ is larger or smaller compared to ${\supp{p(y_t|z_t)}}$ then only a few particles in (\ref{eq:Chap1E37}) are assigned higher weights. This is due to insufficient number of particles in the overlapping region (see Figure \ref{figure:Chap1F1}). As discussed in \cite{R2004}, a standard SIR filter is inefficient in handling such situations. This is because in an SIR filter, the particles from the ISF are generated without taking the current measurement into consideration (see (\ref{eq:Chap1E7})). Methods such as ASIR filter \cite{C2005,LW2001,P1999}; progressive correction \cite{OM2000}; and bridging densities \cite{CG2001} make use of current measurements to allow sampling from high-likelihood regions. Proposition \ref{proposition:Chap1P1} provides an optimal tuning rule for controlling the kernel width and for making an SIR filter adaptive and efficient for different values of $\Omega_t\in\mathbb{R}_+$, where: ${\Omega_t\triangleq\tr[\mathbb{V}_{p(Z_t|Y_{1:t-1})}[Z_t]]/\tr[\mathbb{V}_{p(Y_t|Z_{t})}[Y_t]]}$; $\mathbb{R}_+:=[0,\infty)$; and $\tr[\cdot]$ is the trace operator.
\begin{figure*}[t]
   \centering
   \includegraphics[scale=0.55]{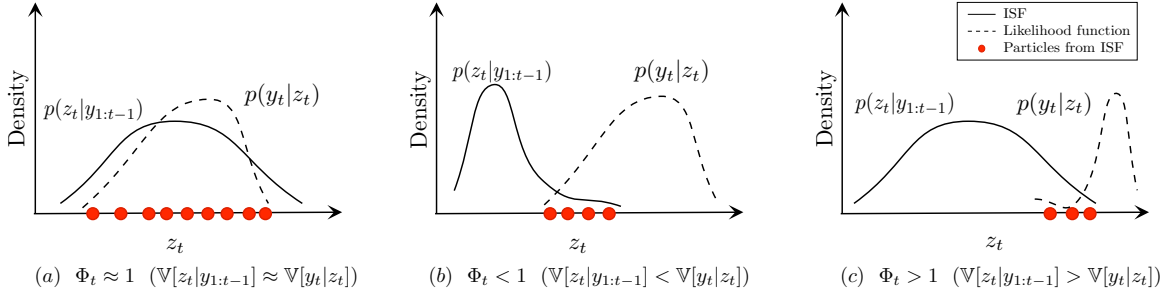}
   \caption{{A schematic diagram to highlight the possible scenarios for different values of $\Omega_t\in\mathbb{R}_+$, where $\Omega_t\triangleq\tr[\mathbb{V}_{p(Z_t|Y_{1:t-1})}[Z_t]]/\tr[\mathbb{V}_{p(Y_t|Z_{t})}[Y_t]]$ and $\tr[\cdot]$ is the trace operator. In Case (a), when $\Omega_t\approx1$, the ISF is mapped in the high likelihood region, which represents an ideal estimation scenario for SIR filters. In Cases (b) and (c), either the ISF is peaked $(\Omega_t<1)$ or the likelihood function is peaked $(\Omega_t>1)$ compared to the other distribution, such that only few number of particles generated from the ISF falls in the likelihood region.}}
   \label{figure:Chap1F1}
\end{figure*}
\begin{proposition}
\label{proposition:Chap1P1}
An optimal tuning for $h_t$ at ${t\in\mathbb{N}}$ based on minimization of the KL divergence between the ISF $p(z_t|y_{1:t-1})$ and target posterior density $p(z_t|y_{1:t})$ is given by
\begin{align}
h^\star_{t}=\argmin_{h_{t}\in[0,1]}\left[-\sum_{i=1}^{N} W^i_{t|t-1}\log[{W}^i_{t|t}]\right],
\end{align}
where: $h^\star_{t}$ is the optimal kernel parameter at $t\in\mathbb{N}$; and ${\{W^i_{t|t-1}\}_{i=1}^N}$ and ${\{{W}^i_{t|t}\}_{i=1}^N}$ are the particle weights given in (\ref{eq:Chap1E35}) and (\ref{eq:Chap1E36}), respectively.
\end{proposition}
\begin{proof}
The KL divergence between $p(z_t|y_{1:t-1})$ and $p(z_t|y_{1:t})$ at $t\in\mathbb{N}$ is given by
\begin{align}
\label{eq:Chap1E38}
D_{q||p}(t)=\int_{\mathcal{Z}}\log\left[\frac{p(z_t|y_{1:t-1})}{p(z_t|y_{1:t})}\right]p(dz_t|y_{1:t-1}),
\end{align}
where $D_{q||p}(t)$ is the KL divergence at ${t\in\mathbb{N}}$. Substituting (\ref{eq:Chap1E8}) into (\ref{eq:Chap1E38}) yields
\begin{subequations}
\begin{align}
D_{q||p}(t)=&\int_{\mathcal{Z}}\log\left[\frac{p(y_t|y_{1:t-1})}{p(y_t|z_t)}\right]p(dz_t|y_{1:t-1}),\\
=&\int_{\mathcal{Z}}\log\left[\frac{\int_{\mathcal{Z}}p(y_t|z_t)p(dz_t|y_{1:t-1})}{p(y_t|z_t)}\right]p(dz_t|y_{1:t-1})\label{eq:Chap1E39}.
\end{align}
\end{subequations}
Computing (\ref{eq:Chap1E39}) in closed form is non-trivial for the model considered in (\ref{eq:Chap1E1}); however, substituting (\ref{eq:Chap1E35}) into (\ref{eq:Chap1E39}) yields an SMC approximation of (\ref{eq:Chap1E39}), such that 
\begin{subequations}
\begin{align}
\widehat{D}_{q||p}(h_{t})=&\int_{\mathcal{Z}}\log\left[\frac{\int_{\mathcal{Z}}p(y_t|z_t)\sum_{j=1}^{N} W^i_{t|t-1}\delta_{{Z}^i_{t|t-1}}(d{z}_{t})}{p(y_t|z_t)}\right]\sum_{i=1}^{N} W^i_{t|t-1}\delta_{{Z}^i_{t|t-1}}(d{z}_{t}),\\
=&\sum_{i=1}^{N} W^i_{t|t-1}\log\left[\frac{\sum_{j=1}^{N}W^i_{t|t-1}p(y_t|Z^i_{t|t-1})}{p(y_t|Z^i_{t|t-1})}\right],\label{eq:Chap1E40}
\end{align}
\end{subequations}
where $\widehat{D}_{q||p}(h_{t})$ is an SMC estimate of $D_{q||p}(t)$. Note that the dependence of $\widehat{D}_{q||p}(h_{t})$ on ${h_{t}}$ can be established from (\ref{eq:Chap1E31}) and (\ref{eq:Chap1E34}). Several algebraic manipulations in (\ref{eq:Chap1E40}) followed by substituting (\ref{eq:Chap1E37}) into (\ref{eq:Chap1E40}) yields
\begin{align}
\label{eq:Chap1E41}
\widehat{D}_{q||p}(h_{t})=&-\sum_{i=1}^{N} W^i_{t|t-1}\log\left[\frac{W^i_{t|t}}{W^i_{t|t-1}}\right].
\end{align}
Finally, a constrained optimization problem can be formulated based on minimization of $\widehat{D}_{q||p}(h_{t})$ with respect to $h_t$, such that
\begin{align}
\label{eq:Chap1E42}
h^\star_{t}=\argmin_{h_{t}\in[0,1]}\widehat{D}_{q||p}(h_{t}).
\end{align}
Substituting (\ref{eq:Chap1E41}) into (\ref{eq:Chap1E42}) yields
\begin{subequations}
\begin{align}
h^\star_{t}=&\argmin_{h_{t}\in[0,1]}\left[-\sum_{i=1}^{N} W^i_{t|t-1}\log\left[\frac{W^i_{t|t}}{W^i_{t|t-1}}\right]\right],\\
=&\argmin_{h_{t}\in[0,1]}\left[-\sum_{i=1}^{N} W^i_{t|t-1}\log\left[W^i_{t|t}\right]\right],\label{eq:Chap1E43}
\end{align}
\end{subequations}
where (\ref{eq:Chap1E43}) follows from the fact that $\sum_{i=1}^{N}$ $W^i_{t|t-1}\log\left[W^i_{t|t-1}\right]$ is independent of $h_t$, which completes the proof.
\end{proof}
\noindent
\begin{rmk}
\label{remark:Chap1RRR}
Proposition \ref{proposition:Chap1P1} provides an optimal tuning rule for (a) correcting over-dispersion in ISF and; (b) making Ad-SIR filter efficient for different values of $\Omega_t\in\mathbb{R}_+$. Note that other tuning rules for ${h_t\in[0,1]}$ can also be readily used in place of Proposition \ref{proposition:Chap1P1}, provided, it is compatible with the developments of previous sections.
\end{rmk}
\subsection{Resampling}
\label{sec:Chap1S4.4}
In importance sampling, degeneracy is a very common problem; wherein, after a few sampling time instances,  the distribution of the weights in (\ref{eq:Chap1E36}) becomes skewed. As a result, the variance of the weights in (\ref{eq:Chap1E36}) increases over time \cite{D2001}; thereby, requiring a large computational effort to update the particles, whose contributions are negligible. See \cite{ZC2003,R2004} for further details. A systematic resampling scheme \cite{KG1996} is adopted here that eliminates the low weighted particles by replacing them with particles with large weight. The choice of systematic resampling is supported by an easy implementation procedure and a lower order of computational complexity $\mathcal{O}(N)$ \cite{A2002}. A systematic resampling step involves drawing $N$ new particles ${\{{Z}^i_{t|t}\}_{i=1}^N}$, with replacement from a set of particles ${\{{Z}^i_{t|t-1}\}_{i=1}^N}$ realized from the ISF, such that the following equality holds
\begin{align}
\label{eq:Chap1E44}
&\pr({Z}^i_{t|t}={Z}^i_{t|t-1})=W^i_{t|t}
\end{align}
for all ${1\leq i\leq N}$. Here $\pr(\cdot)$ is the probability measure. The resampled particles ${\{{Z}^i_{t|t}\}_{i=1}^N\sim p(z_t|y_{1:t})}$ are identically distributed with weights reset to ${\{W^i_{t|t}=N^{-1}\}_{i=1}^N}$. 
\begin{rmk}
A key feature of the resampling step in (\ref{eq:Chap1E44}) is that it takes an independent set of particles ${\{{Z}^i_{t|t-1}\}_{i=1}^N}$ and returns a set of dependent particles ${\{{Z}^i_{t|t}\}_{i=1}^N}$. This is due to the large number of replications of highly weighted particles. As discussed in \cite{S2011}, using correlated particles ${\{{Z}^i_{t|t};~W^i_{t|t}=N^{-1}\}_{i=1}^N}$ in (\ref{eq:Chap1E36}) further degrades the accuracy of the MMSE point estimate computed in Remark \ref{remark:Chap1R2}. In \cite{BN2000}, the authors showed that the rate of convergence of the MMSE point estimates to the true posterior mean decreases as correlation in ${\{{Z}^i_{t|t}\}_{i=1}^N}$ increases. To avoid any performance degradation, the MMSE point estimates are computed before the resampling step.
\end{rmk} 
\begin{rmk}
Stratified \cite{KG1996,LC1998} or residual \cite{LC1998} resampling can also be used as an alternative to the systematic resampling used here. See \cite{ZC2003} for other resampling methods.
\end{rmk}
\section{Missing measurements}
\label{sec:Chap1S4.5}
Missing measurements are common in the process industries, where measurements may not become available at all sampling time instants. An approach to allow Bayesian state-parameter estimation with real-time missing measurements is presented in this section.

From (\ref{eq:Chap1E37}) it is clear that if ${\{Y_t={y}_t\}}$ at ${t\in\mathbb{N}}$ is missing then (\ref{eq:Chap1E37}) can no longer be used to compute (\ref{eq:Chap1E36}) or the MMSE estimates obtained therefrom. To address this, if ${\{Y_t={y}_t\}}$ at ${t\in\mathbb{N}}$ is missing then the ISF ${p(z_t|y_{1:t-1})}$ in (\ref{eq:Chap1E7}) is used instead to compute a one-step ahead predicted MMSE point estimate for the states and parameters at ${t\in\mathbb{N}}$. The procedure to obtain an MMSE estimate under missing measurements is outlined next.
\begin{rmk}
\label{remark:Chap1R5}
Let the SMC approximation of the ISF $p(dz_t|y_{1:t-1})$ be represented by (\ref{eq:Chap1E35}) then a one-step ahead predicted MMSE point estimate for the states and parameters at ${t\in\mathbb{N}}$ can be computed as $\widehat{Z}_{t|t-1}\triangleq\int_{\mathcal{Z}}z_tp(dz_t|y_{1:t-1})\approx\sum_{i=1}^{N}W^i_{t|t-1}Z^i_{t|t-1}$.
\end{rmk}
It is important to note that if ${\{Y_t={y}_t\}}$ at ${t\in\mathbb{N}}$ is missing then the posterior $p(z_t|y_{1:t})$ or its KL divergence with $p(z_t|y_{1:t-1})$ at ${t\in\mathbb{N}}$ cannot be computed either. In other words, $h_{t}$ cannot be optimally tuned (based on Proposition \ref{proposition:Chap1P1}) under missing measurements.  

Note that with Proposition \ref{proposition:Chap1P1}, optimal tuning for $h_{t}$ under missing measurement is not necessary. This is because tuning $h_t$ according to Proposition \ref{proposition:Chap1P1} corrects the variance inflation problem in the SMC approximation of $p(z_t|y_{1:t-1})$ and also projects the particles from it onto the region of high posterior density $p(z_t|y_{1:t})$ (see Remark \ref{remark:Chap1RRR}); however, if $p(z_t|y_{1:t})$ is unavailable at ${t\in\mathbb{N}}$, Proposition \ref{proposition:Chap1P1} only addresses the variance inflation in the SMC approximation of $p(z_t|y_{1:t-1})$, which can be corrected with any $h_t\in[0,1]$ value. 
\begin{rmk}
As a general rule, if ${\{Y_t={y}_t\}}$ at ${t\in\mathbb{N}}$ is missing,  $h_t$ will be assigned its previous optimal value $h^{\star}_{t-1}$. Note that, if necessary, the user can choose any ${h_t \in[0,1]}$ value, or can optimize it based on other tuning rules as well (see Remark \ref{remark:Chap1RRR}).
\end{rmk}
After computing the one-step ahead predicted MMSE state-parameter point estimate at $t\in\mathbb{N}$ (see Remark \ref{remark:Chap1R5}), the Law of Total Probability on $p({z}_{t}| {y}_{1:t-1})$  yields
\begin{align}
\label{eq:Chap1E45}
p({z}_{t+1}| {y}_{1:t-1})=\int_{\mathcal{Z}}p({z}_{t+1}| {z}_{t})p(d{z}_{t}| {y}_{1:t-1}),
\end{align}
where ${p({z}_{t+1}| {y}_{1:t-1})}$ is a two-step ahead prior density, and also the ISF for the sampling time ${t+1}$ under missing ${\{Y_t={y}_t\}}$. Since (\ref{eq:Chap1E45}) does not have a closed form solution, an SMC approximation of it can be obtained by substituting (\ref{eq:Chap1E35}) into (\ref{eq:Chap1E45}), such that
\begin{align}
\label{eq:Chap1E46}
\tilde{p}({z}_{t+1}| {y}_{1:t-1})=&\sum_{i=1}^{N} W^i_{t|t-1}p({z}_{t+1}|{Z}^i_{t|t-1}).
\end{align}
To correct the variance inflation in (\ref{eq:Chap1E46}), kernel smoothing discussed in Section \ref{sec:Chap1S4.2} is applied, such that with kernel smoothing the ISF can now be approximated as follows
\begin{align}
\label{eq:Chap1E47}
\tilde{p}(z_{t+1}| y_{1:t-1})=&\sum_{i=1}^{N} W^i_{t|t-1}p({z}_{t+1}| {\tilde{Z}}^i_{t|t-1}),
\end{align}
where ${\{\tilde{Z}^i_{t|t-1}\}_{i=1}^N=\{X^i_{t|t-1};~\tilde{\theta}^i_{t|t-1}\}_{i=1}^N}$, and
\begin{align}
\label{eq:Chap1EE34}
\tilde{\theta}^i_{t|t-1}= \sqrt{1-h_{t+1}^2}~{\theta}^i_{t|t-1} +(1-\sqrt{1-h_{t+1}^2})~\widehat{\theta}_{t|t-1}.
\end{align} 
In (\ref{eq:Chap1EE34}), $h_{t+1}$ can be tuned based on Proposition \ref{proposition:Chap1P1}, using the next available measurement ${\{Y_{t+1}=y_{t+1}\}}$. Note that from (\ref{eq:Chap1E47}), random particles can be generated by passing ${\tilde{Z}}^i_{t|t-1}$ through $p({z}_{t+1}| {\tilde{Z}}^i_{t|t-1})$ for all ${1\leq i\leq N}$. Using the set of generated random particles, the ISF distribution $p(d{z}_{t+1}| {y}_{1:t-1})$ can be represented as
\begin{align}
\label{eq:Chap1E48}
\tilde{p}(d{z}_{t+1}| {y}_{1:t-1})=&\sum_{i=1}^{N} W^i_{t+1|t-1}\delta_{{Z}^i_{t+1|t-1}}({d{z}}_{t+1}),
\end{align}
where ${\{{Z}^i_{t+1|t-1};~W^i_{t+1|t-1}=w^i_{t|t-1}\}_{i=1}^N}$ is a set of $N$ random particles from (\ref{eq:Chap1E47}). 

Finally, using the next available measurement ${\{Y_{t+1}=y_{t+1}\}}$, the posterior distribution $p(dz_{t+1}| y_{1:t-1},y_{t+1})$ at ${t+1}$ can be approximated using SMC methods, such that
\begin{align}
\label{eq:Chap1E49}
\tilde{p}(d{z}_{t+1}| {y}_{1:t-1},{y}_{t+1})=& \sum_{i=1}^{N} W^i_{t+1|t+1}\delta_{{Z}^i_{t+1|t-1}}(d{z}_{t+1}),
\end{align}
where $\{W^i_{t+1|t+1}\}_{i=1}^N$ are computed using (\ref{eq:Chap1E37}).
\begin{rmk}
\label{remark:Chap1R4}
The on-line Bayesian state-parameter estimation method presented in this section assumes that measurements are missing at random time instants. Note that, the proposed method can also handle cases with multiple consecutively missed measurements. \end{rmk}
\section{On-line estimation algorithm}
\label{sec:Chap1S5}
Algorithms \ref{algorithm:Chap1A1} and \ref{algorithm:Chap1A2} outlines the procedure for estimating ${Z_t\in\mathcal{Z}}$  in (\ref{eq:Chap1E1}) for complete and missing measurements, respectively. Convergence of these algorithms is discussed next.
\begin{algorithm}[t]
  \caption{Complete measurements}
  \label{algorithm:Chap1A1}
  \begin{algorithmic}[1]
    \STATE Select a prior pdf ${Z_0\sim p(z_0)}$ for the states and parameters.
    \STATE Generate $N$ independent and identically distributed particles ${\{{Z}^i_{0|-1}\}_{i=1}^N\sim p({z}_0)}$ and set the associated weights to $\{W^i_{0|-1}=N^{-1}\}_{i=1}^N$. Set $t\leftarrow 1$.
    \STATE Sample ${\{Z^i_{t|t-1}\}_{i=1}^N\sim p(z_t|y_{1:t-1})}$ using (\ref{eq:Chap1E33}). Set ${\{W^i_{t|t-1}=N^{-1}\}_{i=1}^N}$.
    \WHILE{$t\in\mathbb{N}$}
    \STATE Use ${\{Y_t=y_t\}}$ and compute the importance weights $\{W^i_{t|t}\}_{i=1}^N$ from (\ref{eq:Chap1E37}).
    \STATE Compute the point estimate $\widehat{Z}_{t|t}$ using the procedure outlined in Remark \ref{remark:Chap1R2}.
    \STATE Resample the particle set ${\{Z^i_{t|t-1}; W^i_{t|t}\}_{i=1}^N}$ with replacement using (\ref{eq:Chap1E44}).
    \STATE Compute $h^{\star}_{t+1}$ using Proposition \ref{proposition:Chap1P1} and generate $\{{\tilde{\theta}}^i_{t|t}\}_{i=1}^N$ using (\ref{eq:Chap1E31}).
    \STATE Sample ${\{Z^i_{t+1|t}\}_{i=1}^N\sim p(z_{t+1}|y_{1:t})}$ using (\ref{eq:Chap1E33}). Set ${\{W^i_{t+1|t}=N^{-1}\}_{i=1}^N}$.
    \STATE Set $t\leftarrow t+1$.
    \ENDWHILE
  \end{algorithmic}
\end{algorithm}
\begin{algorithm}[h]
  \caption{Missing measurements}
  \label{algorithm:Chap1A2}
  \begin{algorithmic}[1]
    \STATE Select a prior pdf ${Z_0\sim p(z_0)}$ for the states and parameters.
    \STATE Generate $N$ independent and identically distributed particles ${\{{Z}^i_{0|-1}\}_{i=1}^N\sim p({z}_0)}$ and set the associated weights to $\{W^i_{0|-1}=N^{-1}\}_{i=1}^N$. Set $t\leftarrow 1$.
      \STATE Sample ${\{Z^i_{t|t-1}\}_{i=1}^N\sim p(z_t|y_{1:t-1})}$ using (\ref{eq:Chap1E33}). Set ${\{W^i_{t|t-1}=N^{-1}\}_{i=1}^N}$.
    \WHILE{$t\in\mathbb{N}$}
    \IF{${\{Y_t=y_t\}}$ is available}
    \STATE Use ${\{Y_t=y_t\}}$ and compute the importance weights $\{W^i_{t|t}\}_{i=1}^N$ from (\ref{eq:Chap1E37}).
    \STATE Compute the point estimate $\widehat{Z}_{t|t}$ using the procedure outlined in Remark \ref{remark:Chap1R2}.
    \STATE Resample the particle set ${\{Z^i_{t|t-1}; W^i_{t|t}\}_{i=1}^N}$ with replacement using (\ref{eq:Chap1E44}).
    \ENDIF
    \IF{${\{Y_t=y_t\}}$ is unavailable}
    \STATE Compute the predicted point estimate $\widehat{Z}_{t|t-1}$ using the procedure in Remark \ref{remark:Chap1R5}.
    \ENDIF
    \IF{${\{Y_{t+1}=y_{t+1}\}}$ is available}
    \STATE Compute $h^{\star}_{t+1}$ using Proposition \ref{proposition:Chap1P1} and generate $\{{\tilde{\theta}}^i_{t|t}\}_{i=1}^N$ using (\ref{eq:Chap1E31}).
    \STATE Sample ${\{Z^i_{t+1|t}\}_{i=1}^N\sim p(z_{t+1}|y_{1:t})}$ using (\ref{eq:Chap1E33}). Set ${\{W^i_{t+1|t}=N^{-1}\}_{i=1}^N}$.
    \ENDIF
    \IF{${\{Y_{t+1}=y_{t+1}\}}$ is unavailable}
    \STATE Set $h^{\star}_{t+1}\leftarrow h^{\star}_{t}$ and generate $\{{\tilde{\theta}}^i_{t|t-1}\}_{i=1}^N$ using (\ref{eq:Chap1EE34}).
    \STATE Sample ${\{Z^i_{t+1|t-1}\}_{i=1}^N\sim p(z_{t+1}|y_{1:t-1})}$ using (\ref{eq:Chap1E47}). Set $\{W^i_{t+1|t-1}=w^i_{t|t-1}\}_{i=1}^N$.
    \ENDIF
    \STATE Set $t\leftarrow t+1$.
    \ENDWHILE
  \end{algorithmic}
\end{algorithm}
\section{Convergence}
\label{sec:Chap1S6}
Computing the conditional mean of ${Z_t|(Y_{1:t}=y_{1:t})\sim p(z_t|y_{1:t})}$ requires evaluating the multi-dimensional integral over ${\mathcal{Z}}$. As stated earlier, obtaining an analytical solution to the MMSE estimate is not possible for the model considered in (\ref{eq:Chap1E1}). Algorithms {\ref{algorithm:Chap1A1} and \ref{algorithm:Chap1A2} deliver an $N$-particle approximation to the MMSE estimates. Establishing theoretical convergence for Algorithms \ref{algorithm:Chap1A1} and \ref{algorithm:Chap1A2} is beyond the scope of this paper; however, some of the practical issues affecting their convergence, include:
\begin{list}{\labelitemi}{\leftmargin=1em}
\item Finding an optimal ${N<\infty}$, for which the $N$-particle MMSE estimate $\widehat{Z}^N_{t|t}$ would converge to true MMSE estimate ${Z}^\star_{t|t}$ in a ball of some predefined radius is non-trivial; however, note that the estimates can be made accurate for sufficiently large $N$.
\item Inaccurate noise model can prevent the estimates from converging to their true values. To circumvent this problem the noise models are known in their distribution class and their parameters estimated along with model parameters (see Assumption \ref{assumption:Chap1A1}).
\item  Poor choice of ${Z_0\sim p(z_0)}$ can cause serious convergence issues. The problem is particularly severe while estimating the discrete states of hybrid systems. Any discrete change in the state require an adaptive mechanism for redefining the ISF for the states. Since estimation in hybrid systems is not included in the scope of this paper, it will not be considered here. Consideration will be made in selecting $p(z_0)$ in Section \ref{sec:Chap1S8}.
\end{list}
The procedure to reduce computational complexity of Algorithms \ref{algorithm:Chap1A1} and \ref{algorithm:Chap1A2} is discussed next.
\begin{rmk}
\label{remark:Chap1R6}
Algorithms \ref{algorithm:Chap1A1} and \ref{algorithm:Chap1A2} compute an estimate of ${Z_t\in\mathcal{Z}}$. Note for time-invariant systems, estimation of $\theta_t$ can be bypassed if ${\exists t_\alpha \in\mathbb{N}}$, ${\lim_{N\to +\infty}\widehat{\theta}^N_{t|t}-\theta^{\star}=0~\forall t\geq t_\alpha}$, where ${\theta^\star\in\Theta}$ is a vector of true system parameters. The rationale behind this approach is to reduce the computational complexity of  Algorithms \ref{algorithm:Chap1A1} and \ref{algorithm:Chap1A2}  by simply selecting ${\widehat{\theta}_{t|t}=\widehat{\theta}_{t_\alpha|t_\alpha}~ \forall t\geq t_\alpha}$. Caution is required while estimating in a time-varying systems.
\end{rmk}
In the next section, some of the key features of the on-line estimation algorithm presented in this paper are compared against that of an off-line parameter estimation algorithm.
\section{Comparison with off-line algorithm}
\label{sec:Chap1S7}
In processes, where developing an efficient off-line parameter estimator is required, an EM algorithm has been very successful. The EM algorithm is a popular off-line ML based method for parameter estimation in non-linear SSMs with non-Gaussian noise. The key advantage with EM is that it can be adopted under a variety of industry relevant situations. In \cite{S2010,S2011}, the authors used the off-line EM algorithm to estimate the process and noise model parameters ({e.g.,} mean and covariance) under complete measurements. Extension of the EM algorithm for estimation under missing measurements was considered in \cite{G2008}. 

In terms of computational complexity, the particle smoothing step in EM requires $\mathcal{O}(N^2Tn)$ calculations at each iteration \cite{S2010,G2008,S2011}, where $n$ is the state dimension and $T$ is the total number of measurements. Smoothing step with computational complexity $\mathcal{O}(NTn)$ has also appeared \cite{DGMO2011}.  This highlights the scalability issues with the EM algorithm when $n$ is large. The brute-force optimization in the M step of EM further adds to the computational cost. From a theoretical perspective, EM has an advantage in terms of asymptotic efficiency and consistency; however, in practice, solving the maximization step of EM can be prohibitive, especially in large dimensional dynamical systems with long measurement sequence. Depending on the dimension of the system, the number of particles and samples used, the algorithm may take hours to run on a state-of-the art desktop computer \cite{G2008}.

Focussing only on the parameter estimation aspect of Algorithms \ref{algorithm:Chap1A1} and \ref{algorithm:Chap1A2}, the developed method can estimate the process and noise model parameters in real-time with either complete or missing measurement set. The efficacy of the proposed method in dealing with these cases is demonstrated in Section \ref{sec:Chap1S8}. A distinct advantage of the proposed algorithm is that it can also be used for estimating time-varying systems. Computational complexity of Algorithms \ref{algorithm:Chap1A1} and \ref{algorithm:Chap1A2} until time $T$ is of the order $\mathcal{O}(NTs)$ whereas the optimization approach introduced in Proposition \ref{proposition:Chap1P1} has complexity $\mathcal{O}(N)$, where $r$ is the dimension of unknown parameters. Also, by including Remark \ref{remark:Chap1R6}, the computational cost can further be reduced. Direct quantification of the bias introduced through the use of artificial dynamics approach might be difficult as pointed in \cite{Kantas2009}; however, \cite{ABBF2012} proposed the use of PCRLB  for assessing the quality of the parameter estimates. This assessment is done by comparing the MSE for the estimates against the theoretical PCRLB. Experiments in \cite{ABBF2012} have confirmed that using ADA, with the tuning rule in Proposition 1 yields numerically reliable estimates.  
\begin{rmk}
Comparison is not intended to draw conclusions on the validity of the involved algorithms. Instead, it is provided to highlight key features of the Ad-SIR filter in handling situations, which have been considered so far only under off-line settings.
\end{rmk}
\section{Numerical illustrations}
\label{sec:Chap1S8}
In this section, efficacy of Algorithms \ref{algorithm:Chap1A1} and \ref{algorithm:Chap1A2} is illustrated through two numerical examples. The first example is taken from \cite{G2008} and the second example from \cite{S2011}. In this study, the estimation problem is formulated to estimate both states and parameters of a non-linear system, but the analysis is focussed mainly on on-line parameter estimation as it has been less studied compared to the state estimation problem.
\subsection{Example 1: A non-linear and non-Gaussian system}
\label{sec:Chap1S8.1}
Consider the following stochastic SSM \cite{G2008,GA2005}
\begin{subequations}
\label{eq:Chap1E50}
\begin{align}
X_{t+1}&=\alpha_tX_t+\beta_t U_t+V_t,\label{eq:Chap1E50a}\\
Y_t&=\gamma_t\cos X_t+W_t,\label{eq:Chap1E50b}
\end{align}
\end{subequations}
where: ${U_t\sim\mathcal{N}(u_t|0,1)}$; ${V_t\sim\mathcal{N}(v_t|0,Q_t)}$; and ${W_t\sim\mathcal{N}(w_t|0,R_t)}$. The process and measurement noise models in (\ref{eq:Chap1E50a}) and (\ref{eq:Chap1E50b}), respectively, are known in their distribution class and mean, but unknown in their respective variances ${Q_t\in\mathbb{R}_{+}}$ and ${R_t\in\mathbb{R}_{+}}$. \cite{G2008} used this example for off-line estimation of process and noise model parameters under complete and missing measurements using EM algorithm. In this study, real-time state-parameter estimation will be setup using Algorithms \ref{algorithm:Chap1A1} and \ref{algorithm:Chap1A2}.

For comparison with results reported in \cite{G2008}, similar simulation conditions are maintained to the extent possible. As in \cite{G2008}, the initial condition for the true state and true parameters in (\ref{eq:Chap1E50}) are selected as $x^\star_0=1$ and ${\theta}_t^{\star}\triangleq[\alpha_t^{\star}; ~\beta_t^{\star}; ~\gamma_t^{\star};~Q_t^{\star}; ~R_t^{\star}]={[0.9;~1;~1;~0.1;~0.1]~\forall t\in[1,T]}$, respectively.

To estimate ${\theta_t\in\mathbb{R}^5}$, MC simulations are performed using $45$ random realizations of input-output data $\{u_{1:T};~y_{1:T}\}$. For each input-output data set, MMSE estimates ${\widehat{\theta}_{t|t}~\forall t\in[1,T]}$ are computed. For this study a finite filtering time ${T=1000}$ is selected with ${N=20000}$ particles. A large $T$ and $N$ values help reduce variation in ${\widehat{\theta}_{t|t}}$ arising due to randomness in measurement and error associated with SMC approximations, respectively.
\begin{table*}[t]
\caption{{Parameter estimates and standard error computed using Algorithms \ref{algorithm:Chap1A1} and \ref{algorithm:Chap1A2} based on 45 MC simulations.}}
\centering
\begin{tabular}{cc|cccc}
\hline
\multicolumn{1}{c}{Parameter}&\multicolumn{1}{c|}{True}&\multicolumn{4}{c}{Parameter estimates $\pm$ standard deviation $(\widehat{\theta}_{T|T}\pm V^{0.5}_{\theta_T})$}\\
\cline{3-6}
\multicolumn{1}{c}{$\theta_t$}&\multicolumn{1}{c|}{$\theta_t^\star$}&\multicolumn{1}{c|}{0\% Missing}&\multicolumn{1}{c|}{10\% Missing}&\multicolumn{1}{c|}{25\% Missing}&\multicolumn{1}{c}{50\% Missing}\\
\hline
$\alpha_t$&$0.90$&$0.9027\pm0.0060$& $0.9017\pm0.0074$&$0.9014\pm0.0077$&$0.9041\pm0.0079$\\
$\beta_t$&$1.0$&$0.9926\pm0.0210$& $0.9946\pm0.0203$&$0.9913\pm0.0278$&$0.9865\pm0.0367$\\
$\gamma_t$&$1.0$&$1.0179\pm0.0225$&$1.0145\pm0.0208$&$1.0105\pm0.0275$&$0.9743\pm0.0415$\\
$Q_t$&$0.10$&$0.1068\pm0.0124$&$0.1054\pm0.0145$&$0.1037\pm0.0167$&$0.0915\pm0.0197$\\
$R_t$&$0.10$&$0.1068\pm0.0090$&$0.0892\pm0.0076$&$0.0932\pm0.0129$&$0.1101\pm0.0216$\\
\hline
\end{tabular}
\label{table:Chap1T1}
\end{table*}
The prior density ${\theta_0\sim \mathcal{N}(\theta_0|M_\theta,C_\theta)}$ is selected as a mutually independent multi-variate normal distribution with mean ${M_\theta=[0.5;~ 0.5;~ 0.5;~0.2;~0.2]}$ and covariance $C_\theta=\text{diag}([1;~ 1;~ 1;~0.05;~0.05])$, where $\text{diag}(\cdot)$ is a diagonal matrix. 

In this simulation study, estimation is performed on four different experiment runs each with $0\%,~10\%,~25\%$ and $50\%$ randomly missing measurements. A MC based MMSE parameter estimates $\widehat{\theta}_{T|T}$ along with the standard estimation error at sampling time $t=T$ are given in Table \ref{table:Chap1T1}. In each of the four experiments the estimated parameters $\widehat{\theta}_{T|T}$ are in the neighbourhood of $\theta_T^\star$.  Also, comparing with the results reported in \cite{G2008}, the proposed method delivers $\widehat{\theta}_{T|T}$ in the neighbourhood of $\theta_T^\star$ with high statistical reliability. Higher parameter accuracy can be attributed to large $T$ and $N$ values used here in contrast to $T=100$ and $N=150$ used by \cite{G2008}. This highlights the advantage of Ad-SIR filter over EM algorithm; wherein, large $N$ can be used to approximate the posterior without significant increase in the computational load.
\begin{figure}[t!]
   \centering
   \includegraphics[scale=0.73]{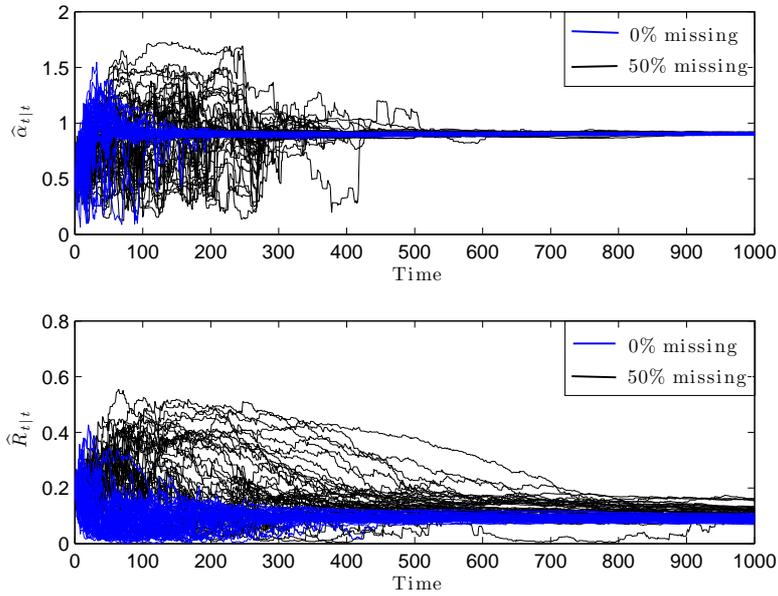}
   \caption{{MMSE estimates of:  [Top] $\widehat{\alpha}_{t|t}$ and [Bottom] $\widehat{R}_{t|t}$ computed using Algorithms \ref{algorithm:Chap1A1} and \ref{algorithm:Chap1A2} based on 45 simulations.}}
   \label{figure:Chap1F2}
\end{figure}
\begin{figure}[h!]
   \centering
   \includegraphics[scale=0.73]{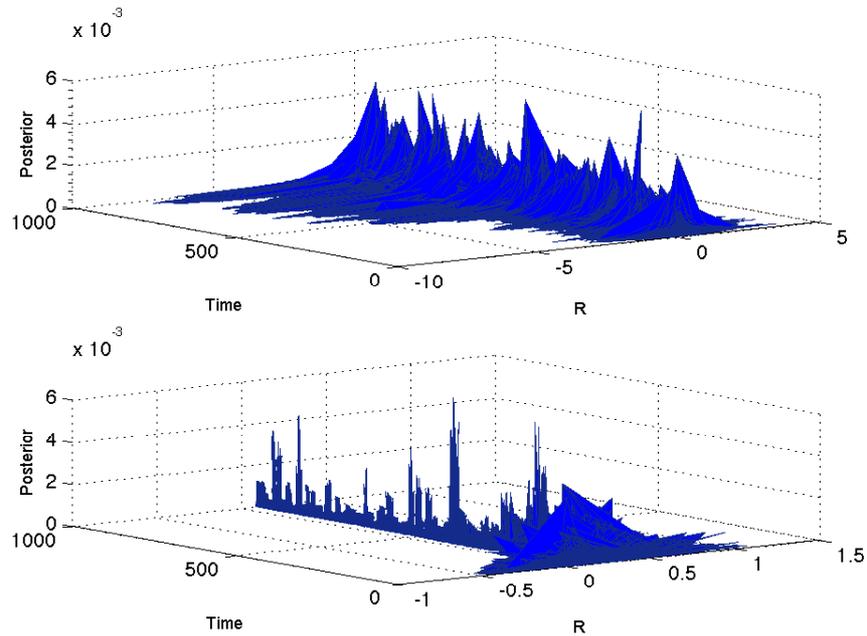}
   \caption{Posterior distribution $\tilde{p}(R_t|y_{1:t})$ $\forall t\in[1,T]$ under $0\%$ missing measurements computed using Algorithm \ref{algorithm:Chap1A1}: [Top] without kernel smoothing method, and [Bottom] with kernel smoothing method and tuning rule selected as Proposition \ref{proposition:Chap1P1}.}
   \label{figure:Chap1F3}
\end{figure}
\begin{figure}[h!]
   \centering
   \includegraphics[scale=0.73]{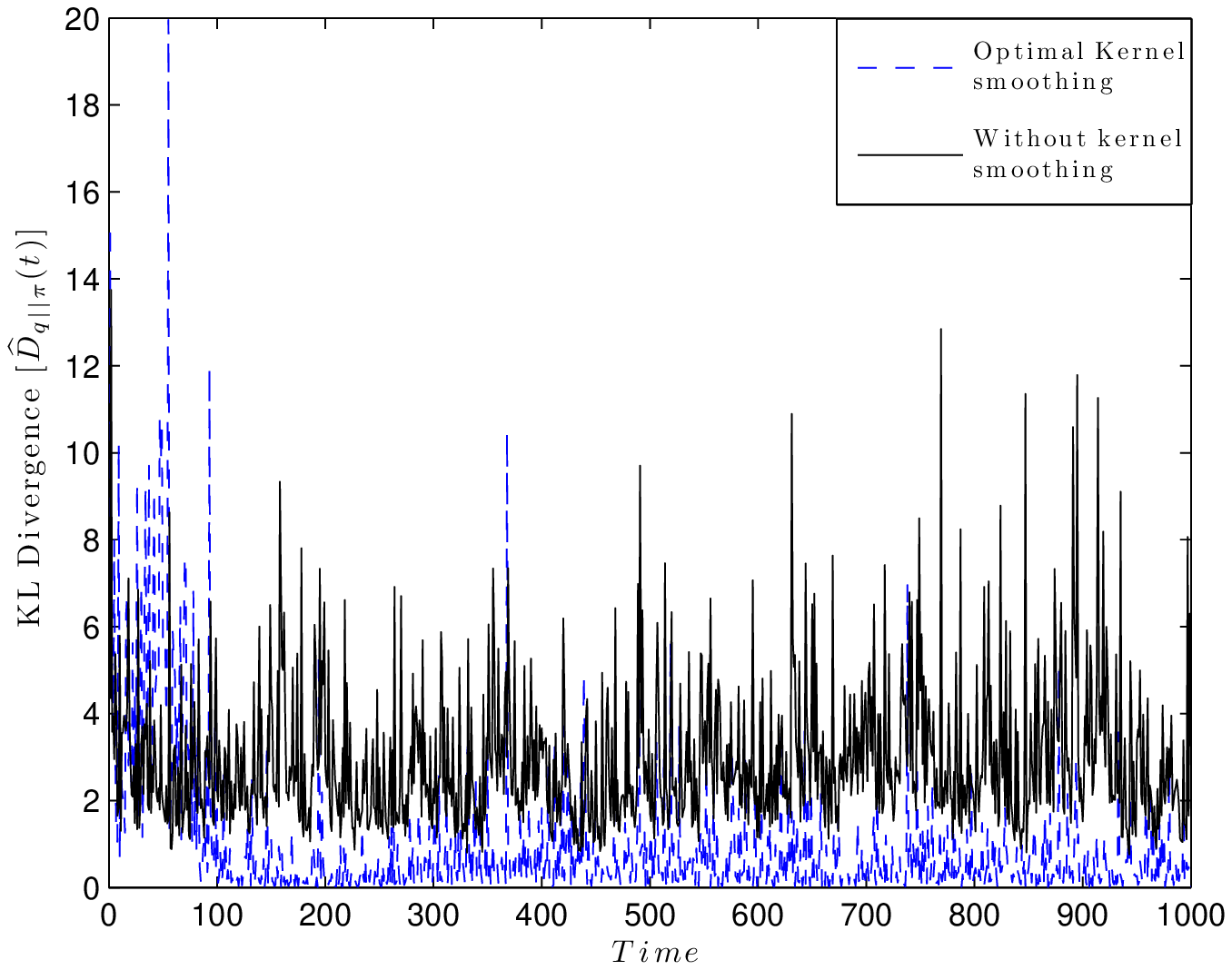}
   \caption{{KL divergence between $\tilde{p}(z_t|y_{1:t-1})$ and $\tilde{p}(z_t|y_{1:t})~\forall t\in[1,T]$ computed using Algorithm \ref{algorithm:Chap1A1}. The divergence is computed with $T=1000$ and $N=20000$.}}
   \label{figure:Chap1F4}
\end{figure}

Figure \ref{figure:Chap1F2} shows the MMSE estimates $\widehat{\alpha}_{t|t}$ and $\widehat{R}_{t|t}~\forall t$ $\in[1,T]$ computed using Algorithm \ref{algorithm:Chap1A1} (for $0\%$ missing measurements) and Algorithm \ref{algorithm:Chap1A2} (for $50\%$ missing measurements). Under $0\%$ missing measurements, the estimates converge in the neighbourhood of $\theta_{T}^\star$ within a few sampling time instants; whereas, as the percentage of missing measurements increases to $50\%$, the estimates take longer to convergence.

Computation of $\widehat{\theta}_{t|t}~\forall t\in[1,T]$ took $210$ seconds (for $0\%$ missing measurements) on a 3.33 GHz Intel Core i5 processor running on Windows 7. Computation under missing measurements is even faster, as the optimization step for tuning the kernel parameter is not required at all sampling time instants.

Figure \ref{figure:Chap1F3}[Top] validates the comment made in Remark \ref{remark:Chap1RE3} that without correcting the inflation problem, SMC based marginalized posterior density estimate would continue to disperse over time. The advantage of using the kernel smoothing method with Proposition \ref{proposition:Chap1P1} is evident from Figure \ref{figure:Chap1F3}[Bottom]; wherein, the proposed method not only corrects dispersion in the marginalized posterior density, but also reduces it substantially around the estimates. In Figure \ref{figure:Chap1F4}, KL divergence between $\tilde{p}(z_t|y_{1:t-1})$ and $\tilde{p}(z_t|y_{1:t})$ is shown. Comparing the mean and  variance of the two trajectories in Figure \ref{figure:Chap1F4} it is clear that Proposition \ref{proposition:Chap1P1} significantly reduces divergence between the ISF and posterior density.

In summary, Figures \ref{figure:Chap1F2} through \ref{figure:Chap1F4} validate the usefulness of Proposition \ref{proposition:Chap1P1} in achieving convergence of $\widehat{\theta}_{T|T}$ in the neighbourhood of $\theta_T^\star$ under compete and missing measurements. Another non-linear and non-Gaussian example is considered next.
\subsection{Example 2: A non-linear and non-Gaussian system}
\label{sec:Chap1S8.2}
In Section \ref{sec:Chap1S8.1}, efficacy of Algorithms \ref{algorithm:Chap1A1} and \ref{algorithm:Chap1A2} was established under different percentage of missing measurements. In this study, estimation capability of Algorithm \ref{algorithm:Chap1A1} is demonstrated for different values of ${\Gamma_t \in\mathbb{R}_+}$, where ${\Gamma_t\triangleq\mathbb{V}_{p(Z_t|Z_{t-1})}[Z_t]/\mathbb{V}_{p(Y_t|Z_t)}[Y_t]}$. Consider the following discrete-time, stochastic non-linear autonomous SSM \cite{D2001,S2011}
\begin{subequations}
\label{eq:Chap1E51}
\begin{align}
X_{t+1}&=\frac{X_t}{\alpha_t}+\frac{\beta_tX_t}{1+X^2_t}+\kappa_t\cos(1.2t)+V_t,\\
Y_t&=\gamma_tX^2_t+W_t,
\end{align}
\end{subequations}
where: ${V_t\sim\mathcal{N}(v_t|0,Q_t)}$; and ${W_t\sim\mathcal{N}(w_t|0,R_t)}$. The true initial state is chosen as $x_0^\star=5$ and the true parameters are selected as $\theta_t^\star\triangleq [\alpha^\star_t;~\beta^\star_t;~\kappa^\star_t;~\gamma^\star_t;~Q^\star_t;~R^\star_t]=[2.0;~25;~8.0;~0.05;~\{0.10;~1.0\};~\{0.10;~1.0\}]~\forall t\in[1,T]$, where $\{\cdot;\cdot\}$ denote a set of possible discrete values for $Q_t$ and $R_t$, considered in this study. In the simulation, the algorithm parameters are selected as $T=100$ seconds and $N=20000$ particles.

On-line estimation of process and noise model parameters in (\ref{eq:Chap1E51}) is considered for three independent cases, with each differing in the choice of $\Gamma_t~\forall t\in[1,T)$. In the first experiment ${\Gamma_t=1}$ (with ${Q_t=0.1;~R_t=0.1)}$ is selected. For the second and third experiment, ${\Gamma_t=0.1}$ (with ${Q_t=0.1;~R_t=1)}$ and  ${\Gamma_t=10}$ (with ${Q_t=1;~R_t=0.1)}$ is selected, respectively. The choice of the experiments denote the cases in Figure \ref{figure:Chap1F1}. 

The prior density ${\theta_0\sim \mathcal{N}(\theta_0|M_\theta,C_\theta)}$ is selected as a mutually independent multi-variate normal distribution with mean ${M_\theta=[1;~20;~10;~1;~0.5;~0.5]}$ and covariance $C_\theta=\text{diag}([1;~15;~5;~1;~1;~1])$. Large variance ensures that $\theta_0^\star$ is included in the $\supp p(\theta_0)$.

As in Section \ref{sec:Chap1S8.1}, $45$ MC simulations are performed.  Using Algorithm \ref{algorithm:Chap1A1}, a MC MMSE parameter estimates $\widehat{\theta}_{T|T}$ for the three experiments are given in Table \ref{table:Chap1T2}. Small uncertainties associated with $\widehat{\theta}_{T|T}$ across the range of $\Gamma_t$ values suggest high statistical reliability of the estimates. Moreover, comparing the estimates with the true values it is evident that the estimate $\widehat{\theta}_{T|T}$ is in the neighbourhood of $\theta_{T|T}^\star$.  Algorithm \ref{algorithm:Chap1A1} yields the most reliable estimates for $\Gamma_t=1$. This is because $\Gamma_t=1$ presents an ideal scenario for filtering.

Estimates of $\widehat{\gamma}_{t|t}$ and $\widehat{Q}_{t|t}$ for ${\Gamma_t=10}$ are given in Figure \ref{figure:Chap1F5}. On average, $\widehat{\gamma}_{t|t}$ converges in the neighbourhood of $\gamma^\star_T$ in about $t=10$ seconds, whereas $\widehat{Q}_{t|t}$ takes $t=65$ seconds to converge. For this simulation, computation of ${\widehat{\theta}_{t|t}~\forall t\in[1,T]}$ took $21$ seconds of CPU time to complete. Figure \ref{figure:Chap1F6} gives the kernel parameter computed using Proposition \ref{proposition:Chap1P1}.
\begin{table*}[t]\
\caption{{Parameter estimates and standard error computed using Algorithms \ref{algorithm:Chap1A1} for different $\Gamma_t~\forall t\in[1,T]$ based on $45$ MC simulations.}}
\centering
\begin{tabular}{cc|ccc}
\hline
\multicolumn{1}{c}{Parameter}&\multicolumn{1}{c|}{True}&\multicolumn{3}{c}{Parameter estimates $\pm$ standard deviation $(\widehat{\theta}_{T|T}\pm V^{0.5}_{\theta_T})$}\\
\cline{3-5}
\multicolumn{1}{c}{$\theta_t$}&\multicolumn{1}{c|}{$\theta^\star_t$}&\multicolumn{1}{c|}{$\Gamma_t=1$}&\multicolumn{1}{c|}{$\Gamma_t=0.1$}&\multicolumn{1}{c}{$\Gamma_t=10$}\\
\multicolumn{1}{c}{}&\multicolumn{1}{c|}{}&\multicolumn{1}{c|}{$(Q_t=0.1;~R_t=0.1)$}&\multicolumn{1}{c|}{$(Q_t=0.1;~R_t=1)$}&\multicolumn{1}{c}{$(Q_t=1;~R_t=0.1)$}\\
\hline
$\alpha_t$&$2.0$&$2.0358\pm0.0400$& $2.0694\pm0.0812$&$2.0845\pm0.0791$\\
$\beta_t$&$25$&$24.250\pm1.5273$& $23.686\pm1.5997$&$23.916\pm1.6806$\\
$\kappa_t$&$8.0$&$7.9004\pm0.3873$&$7.7611\pm0.4154$&$7.6728\pm0.5329$\\
$\gamma_t$&$0.05$&$0.0530\pm0.0052$&$0.0557\pm0.0061$&$0.0566\pm0.0067$\\
$Q_t$&$-$&$0.1202\pm0.0154$&$0.1284\pm0.0204$&$0.9144\pm0.1543$\\
$R_t$&$-$&$0.1084\pm0.0151$&$0.9054\pm0.1126$&$0.1072\pm0.0157$\\
\hline
\end{tabular}
\label{table:Chap1T2}
\end{table*}

The advantage of using KL divergence  based tuning rule for $h_t$ is highlighted in Figure \ref{figure:Chap1F7}. Figure \ref{figure:Chap1F7} gives the SMC based approximate marginalized posterior distribution $\tilde{p}(\beta_T|y_{1:T})$ for different choices of $h_t~\forall t\in[1,T]$. It is clear that with the proposed tuning rule, Algorithm \ref{algorithm:Chap1A1} projects $\tilde{p}(\beta_T|y_{1:T})$ around the true parameter $\beta^\star_{T}=25$ (see Table \ref{table:Chap1T2}). 

Interestingly, with $h_t=0.01~\forall t\in[1,T]$, a single particle representation of $\tilde{p}(\beta_T|y_{1:T})$ is obtained (see Figure \ref{figure:Chap1F7}). This is because as ${h_t\rightarrow 0}$, ${\Sigma_{\theta_t}=h^2_tV_{\theta_{t-1}}\rightarrow 0~\forall t \in[1,T]}$. In the limiting case, when ${h_t=0}$, $\beta_t$ has a stationary dynamics. It is well known that using SMC methods in such situations result in parameter sample degeneracy (see Section \ref{sec:Chap1S1}).
\begin{figure}[t]
   \centering
   \includegraphics[scale=0.75]{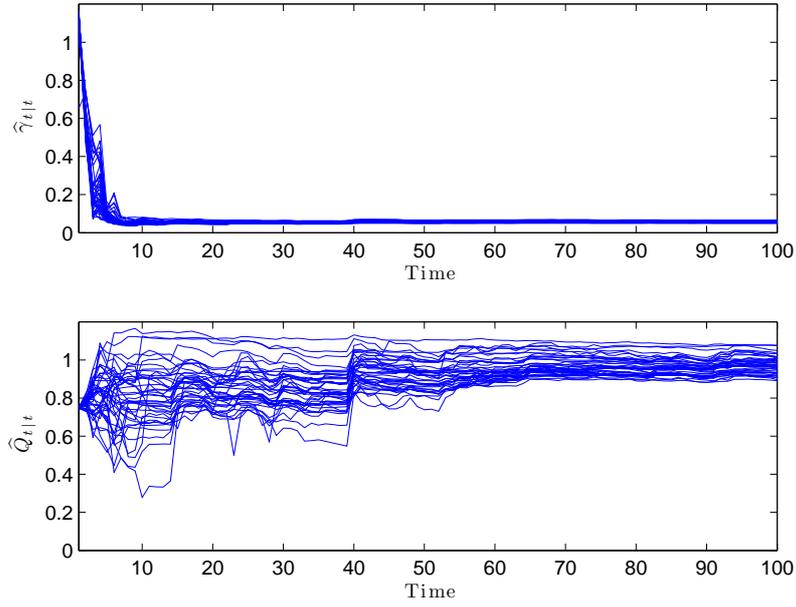}
   \caption{{MMSE estimates of: [Top] $\widehat{\gamma}_{t|t}$ and [Bottom] $\widehat{Q}_{t|t}$ computed using Algorithm \ref{algorithm:Chap1A1} for $\Gamma_t=10~\forall t\in[1,T]$. It is based on $45$ MC simulations with $0\%$ missing measurements.}}
   \label{figure:Chap1F5}
\end{figure}
\begin{figure}[h!]
   \centering
   \includegraphics[scale=0.75]{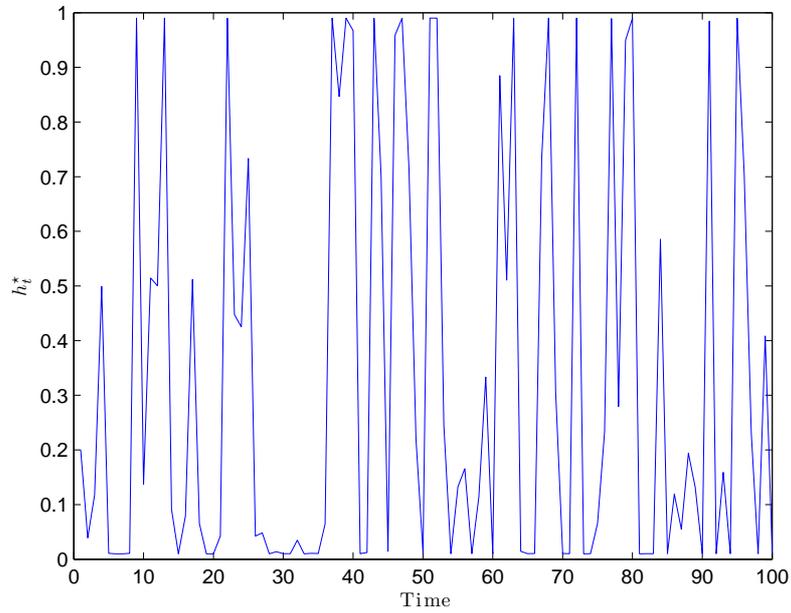}
   \caption{{Optimal kernel $h_t~\forall t\in[1,T]$ tuned using Proposition \ref{proposition:Chap1P1} for $\Gamma_t=10~\forall t\in[1,T]$ and $0\%$ missing measurements.}}
   \label{figure:Chap1F6}
\end{figure}
\begin{figure}[h!]
   \centering
   \includegraphics[scale=0.75]{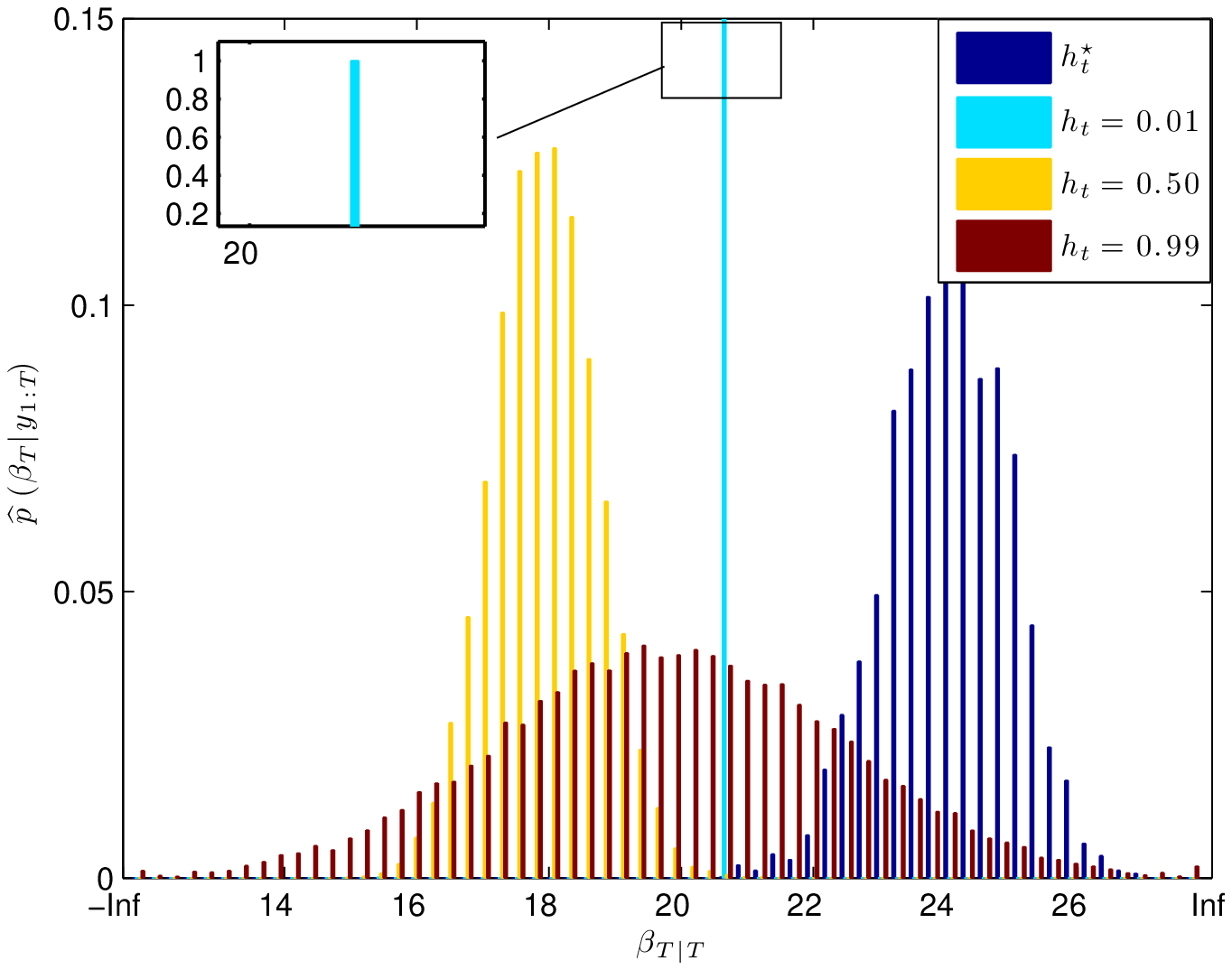}
   \caption{{Approximate marginalized posterior distribution $\tilde{p}(\beta_t|y_{1:t})$ at $t=T$ computed based on different tuning rules for $h_t~\forall t\in[1,T]$. In the graph, $h_t^*$ represents the optimal tuning based on Proposition \ref{proposition:Chap1P1} (see Figure \ref{figure:Chap1F6}).}}
   \label{figure:Chap1F7}
\end{figure}

Studying the other extreme case, with ${h_t=0.99~\forall t\in[1,T]}$ the posterior density ${\tilde{p}(\beta_T|y_{1:T})}$ in Figure \ref{figure:Chap1F7} has a wide support. This can again be understood by analysing ${\tilde{p}(\beta_t|y_{1:t})~\forall t\in[1,T]}$ in limits. As ${h_t\rightarrow 1}$, the set of smoothed particles in (\ref{eq:Chap1E31}) are projected closer to the mean ${{\widehat{\theta}}_{t-1|t-1}}$. Under the limiting case, when ${h_t=1}$ the marginalized ISF is given by ${\tilde{p}({\theta}_t | {y}_{1:t-1})=\sum_{i=1}^{N}W^i_{t-1|t-1}\mathcal{N}(\theta_t|{\widehat{\theta}}_{t-1|t-1}, V_{\theta_{t-1}})}$. Note that, generating particles from $\tilde{p}({\theta}_t | {y}_{1:t-1})$ under the limiting case only depends on the estimated parameter covariance $V_{\theta_{t-1}}$. It is easy to see that in such situations, arbitrarily wide distributions for the SMC based approximate marginalized posterior density can be obtained depending on ${V_{\theta_{t-1}}}$ values.

In summary, this simulation study demonstrates the efficacy of the proposed optimal tuning rule for a range of process to measurement noise variance ratio.
\section{Conclusions}
In this paper, a Bayesian algorithm for on-line state and parameter estimation in discrete-time, stochastic non-linear state-space models is presented. The proposed algorithm uses an adaptive SIR filter to deliver an minimum mean-square error estimate at each filtering time. The extension of the algorithm to handle missing measurements in real-time is also presented. The usual variance inflation problem introduced by adding artificial parameter dynamics is corrected by introducing a kernel smoothing algorithm. An optimal tuning rule for the kernel smoothing parameter is presented under an on-line optimization framework. The usual degeneracy issues with sequential-importance-resampling filter under different process to measurement noise ratios are avoided through the kernel smoothing process based on Kullback-Leibler divergence. The proposed algorithm is an `optimization-free' estimator, which makes it efficient and computationally fast, which is a major advantage over the traditional maximum-likelihood based methods. Finally, the performance of the proposed method was demonstrated on two non-linear simulation examples.
\section*{Acknowledgement}
This work was supported by the Natural Sciences and Engineering Research Council (NSERC), Canada.
\bibliographystyle{IEEETran}
\bibliography{BIB}
 \end{document}